\documentclass[11pt,a4paper]{article}

\usepackage[utf8]{inputenc} 
\usepackage[T1]{fontenc}    %
\usepackage[english]{babel} 
\usepackage[final]{microtype} 

\usepackage{amssymb,amsmath,mathtools} 
\usepackage{amsthm} 

\usepackage{xcolor}
\usepackage{bbm}
\usepackage{braket}

\usepackage[hmargin=0.12\paperwidth,vmargin=0.19\paperwidth,bindingoffset=0cm]%
{geometry} 

\numberwithin{equation}{section} 

\theoremstyle{plain}
  
  \newtheorem{prop}{Proposition}
  \newtheorem{lemma}{Lemma}
  \newtheorem{cor}{Corollary}
\theoremstyle{definition}
  
  \newtheorem{remark}{Remark}
   \newtheorem{con}{Conjecture}
  
  \numberwithin{prop}{section}
   \numberwithin{cor}{section}
   \numberwithin{remark}{section}
  \numberwithin{con}{section}

\title{\Large\bfseries Power spectra of Dyson's circular ensembles}%
   
\author{Peter J. Forrester${}^{(1)}$ and Nicholas S. Witte ${}^{(1,2)}$}
\date{}

\begin{document}

\maketitle

${}^{(1)}$
School of Mathematics and Statistics,  The University of Melbourne,
Victoria 3010, Australia. \: \: Email: {\tt pjforr@unimelb.edu.au}; \\

${}^{(2)}$
School of Mathematics and Statistics, Victoria University of Wellington, Wellington, New Zealand. 
 \: \: Email: {\tt n.s.witte@protonmail.com};

\bigskip

\begin{abstract}
\noindent 
The power spectrum is a Fourier series statistic associated with the covariances of the displacement from average
positions of the members of an eigen-sequence. When this eigen-sequence has rotational invariance,
as for the eigen-angles of Dyson's circular ensembles, recent work of Riser and Kanzieper has uncovered
an exact identity expressing the power spectrum in terms of the generating function for the conditioned gap
probability of having $k=0,\dots,N$ eigenvalues in an interval. These authors moreover showed
how for the circular unitary ensemble integrability properties of the generating function, via a particular Painlev{\'e} VI system,
imply a computational scheme for the corresponding power spectrum, and allow for the determination of its large $N$ limit. 
In the present work, these results are extended to the case of the circular orthogonal ensemble and circular symplectic ensemble,
where the integrability is expressed through four particular Painlev{\'e} VI systems for finite $N$, and two Painlev{\'e} III$'$ systems for the limit $N\to\infty$,
and also via corresponding Fredholm determinants. 
The relation between the limiting power spectrum $S_\infty(\omega)$, where $\omega$ denotes the Fourier variable,
and the limiting generating function for the conditioned gap probabilities is particular direct,
involving just a single integration over the gap endpoint in the latter. Interpreting this generating function as the characteristic function of
a counting statistic  allows for it to be shown that
 $S_\infty(\omega) \mathop{\sim} \limits_{\omega \to 0} {1 \over \pi \beta | \omega|}$, where $\beta$ is the Dyson index.
\end{abstract}

\vspace{3em}

\section{Introduction}
The classic series of papers on the foundations of random matrix theory,
by Dyson, and by Dyson and Mehta in collaboration (see reprints of the originals in
\cite{Po65}, as well as an extended commentary), as motivated by the consideration of
the statistical properties of the spectra of highly excited nuclei, places emphasis on
a small number of statistical quantities. One is the number variance $\Sigma^2(L) = {\rm Var}(
\mathcal N(L))$, where $\mathcal N(L))$ is the random variable specifying the number of 
eigenvalues in an interval of length $L$. Here length is measured in units of the mean
spacing between neighbouring eigenvalues, or equivalently with the latter normalised to unity.
Another is the two level form factor, defined as the Fourier transform of the two-point cluster
function, while a third is the distribution function for the spacing between successive eigenvalues
(in fact this latter statistic was introduced and its study initiated by Wigner; see again \cite{Po65} as
a convenient source for the original papers). Each has had significant impact on a range of 
subsequent applications \cite{BFPW81,Me04,Ha18}, and exhibit sufficiently rich mathematical content
to sustain theoretical investigation into recent years. Such contemporary references include \cite{FL20,SLMS21a,CGMY22,AGL21}
(number variance) \cite{Fo21a,CES23,FKLZ24,CG24} (two level form factor) and \cite{BFM17,Ni21,TRK24,Ni24} (spacing
distribution).

In subsequent times, starting with Odlyzko \cite{Od87} in his 1987 numerical study of the statistical properties
of the Riemann zeros at large height, then in the early 2000's by Rela\~{n}o and collaborators \cite{RGMRF02}
in the context of $1/f$ noise and quantum chaos, and most comprehensively in a series of recent works involving
Kanzieper and Riser  \cite{ROK17,ROK20,RK21,RK23,RTK23}, the power spectrum statistic $S_N(\omega)$
has been added to the above list. For a sequence of eigenvalues parametrised as angles $0 < \theta_1 < \cdots \theta_N < 2 \pi$,
the power spectrum statistic is defined by the Fourier sum
 \begin{equation}\label{1.1}
 S_N(\omega) = {1 \over N \Delta_N^2} \sum_{1 \le l,m\le N} \langle
 \delta \theta_l \delta \theta_m \rangle e^{i \omega ( l - m)}, \qquad \omega \in \mathbb R,
 \end{equation}
where $ \langle\cdot\rangle $ denotes the average with respect to an eigenvalue probability density function.
 Here $\delta \theta_l = \theta_l - \langle \theta_l  \rangle$ is the displacement from the mean $ \langle \theta_l  \rangle$
 and so $[\langle
 \delta \theta_l \delta \theta_m \rangle]_{l,m=1}^N$ is the covariance matrix of the level displacements. The quantity $\Delta_N := 2 \pi/N$ 
 is the mean spacing between eigenvalues. One observes $ S_N(\omega + 2 \pi) =  S_N(\omega) $ and $S_N(-\omega) = S_N(\omega)$,
 so it suffices to restrict $\omega$ to the interval $0 \le \omega \le \pi$.
 
 A primary feature of $S_N(\omega)$ is that it tends to a well defined $N \to \infty$ limit $S_\infty(\omega)$ \cite{ROK20,RK23}.
 Moreover, in \cite{RTK23} $S_\infty(\omega)$ has been shown to relate to the covariances of two nearest neighbour spacings between
 scaled eigenvalues. With regards to the latter, define $x_l = (N / 2 \pi) \theta_l$, $(l=1,2,\dots,)$, so that the average of $x_{l+1} - x_{l}$ is unity.
 Define $s_j(p) = x_{j+1+p} - x_j$ as the spacing between $x_j$ and its $(p+1)$-th neighbour to the right $x_{j+1+p}$, and consider the covariance
 ${\rm cov}_\infty  (s_j(0), s_{j+k}(0))$,\footnote{This covariance can be rewritten according to the identity \cite{BFPW81}
 $$
 {\rm cov}_\infty (s_j(0), s_{j+k+1}(0)) =  {1 \over 2} \Big ({\rm Var}_\infty(s_j(k+1)) - 2 {\rm Var}_\infty(s_j(k)) + {\rm Var}_\infty(s_j(k-1)) \Big ).
 $$
 }
 where the subscript indicates that  the limit $N \to \infty$
 has been taken. Assuming translation invariance, this is independent of $j$. The result of \cite{RTK23} gives that the power spectrum for
 $\{ {\rm cov} \, (s_j(0), s_{j+k}(0))_\infty \}_{k=0,\pm1,\dots}$ relates to $S_\infty(\omega)$ according to
  \begin{equation}\label{1.1x}
  \sum_{k=-\infty}^\infty {\rm cov}_\infty(s_j(0), s_{j+k}(0)) e^{i \omega k} = 4 \sin^2(\omega/2) S_\infty(\omega).
 \end{equation}  
 
 Our interest in the present work is to extend some of the integrability results of \cite{RK23} for the calculation of $S_N(\omega)$ in Dyson's
 circular unitary ensemble (CUE) --- equivalently Haar distributed unitary matrices from the classical group $U(N)$ --- and its  $N \to \infty$ 
 limit, to the two other circular ensembles introduced by Dyson \cite{Dy62}. These are the circular orthogonal ensemble (COE) of symmetric
 unitary matrices, and the  circular symplectic ensemble (CSE) of self-dual quaternion unitary matrices (see \cite[Ch.~2]{Fo10} for more on
 these ensembles). Specifically, integrability structures will be used to provide efficient numerical schemes for the computation of
 $S_N(\omega)$ in these cases, and most significantly its limiting form $S_\infty(\omega)$. 
 Sections \ref{S2.2}, \ref{S2.3} and \ref{S2.4} give the details of our methodology for the computation of 
 $S_N(\omega)$ for the CUE, COE and CSE respectively. Our methodology for the computation of the limiting quantity
 $S_\infty(\omega)$ for
 these ensembles is detailed in Section \ref{S3.2}, and in the COE and CSE cases tabulations are provided. This
 section also contains results on the finite size scaling associated with the convergence of $S_N(\omega)$ to $S_\infty(\omega)$ in the case of the COE.
  
  Beyond numerical considerations, our study contains several analytic properties of $S_N(\omega)$  and its $N \to \infty$ limit $S_\infty(\omega)$, these often
 applying more generally to the circular $\beta$ ensemble. These are the large $N$ asymptotic formula
 for $S_N(0)$ given in (\ref{2.4}), the first two terms of the small $\omega$ asymptotic form of $S_\infty(\omega)$
 given in (\ref{3.2g}), and the first three terms of the large $k$ asymptotic expansion of the covariances
 $ {\rm cov}_\infty  (s_j(0), s_{j+k}(0)) $ from (\ref{1.1x}) given in (\ref{3.9}).  
 An analytic result of a different type, relating to the large $t$ asymptotic form of a pair of sigma Painlev{\'e} III$'$  transcendents
 is given in Conjecture \ref{C3.1}.

 \section{Structured formulas for finite $N$}
 \subsection{Preliminaries}
 Generally the eigenvalue probability density function of the circular $\beta$ ensemble is given by
  \begin{equation}\label{2.0}
  {\mathcal P}_{N,\beta}(\theta_1,\dots,\theta_N) = {1 \over Z_{N,\beta}} \prod_{1 \le j < k \le N} |
  e^{i \theta_k} - e^{i \theta_j} |^\beta,
   \end{equation}
with normalisation
 \begin{equation}\label{2.0a}
 Z_{N,\beta} = (2 \pi)^N {\Gamma(1 + N \beta/2) \over (\Gamma(1 + \beta/2))^N}.
  \end{equation}
  The cases $\beta = 1,2$ and 4 correspond to the COE, CUE and CSE respectively.
  With $E_{N,\beta}(l;(\theta, \theta'))$ denoting the probability that there are exactly $l$
  eigenvalues in the interval $(\theta, \theta')$ of the circular $\beta$ ensemble, the corresponding
  generating function is defined by
   \begin{equation}\label{2.1} 
 \mathcal  E_{N,\beta}((\theta, \theta');\xi)  := \sum_{l=0}^N (1 - \xi)^l E_{N,\beta}(l;(\theta, \theta')).
  \end{equation} 
  
  Riser and Kanzieper \cite{RK23} have shown that the power spectrum statistic (\ref{1.1}) can,
  for the general $\beta > 0$ circular ensemble,
  be expressed in terms of the generating function (\ref{2.1}).
  
  \begin{prop}\label{P2.1} (\cite[Corollary 3.3]{RK23}) Let the power sum statistic for the circular $\beta$ ensemble
  be denoted $S_{N,\beta}(\omega)$. For $q=0,1$ define
  \begin{equation}\label{2.2}  
  I_{N,q}^{(\beta)}(z) := {N \over 2 \pi} \int_0^\pi \Big ( {\phi \over 2 \pi} \Big )^q   \mathcal E_{N,\beta}((0, \phi);1- z) \, d \phi. 
  \end{equation}  
  Setting $z = e^{i \omega}$ we have
    \begin{equation}\label{2.3}  
  S_{N,\beta}(\omega) = {2 z \over (1 - z)^2} \bigg ( - {\rm Re} \Big (  I_{N,0}^{(\beta)}(z)     - (1 - z^{-N}) I_{N,1}^{(\beta)}(z)  \Big ) 
  + {2 \over N} \Big ( {\rm Re} ( z^{-N/2}  I_{N,0}^{(\beta)}(z)) \Big )^2 \bigg ).
  \end{equation}
  \end{prop}
  
  In the case $\beta = 2$, it has been shown in \cite{ROK20,RK23} how this can be used to compute $S_{N,\beta}(\omega)$. Below, we will revise
  the available methodologies, and furthermore show that computation based on (\ref{2.3}) is also possible for $\beta = 1$ and 4. But before doing so,
  we will show that $S_{N,\beta}(\omega)$ at $\omega = 0$ can be expressed in terms of Var$\, {\mathcal N}_{(0,\phi)}^{(\beta)}$, where ${\mathcal N}_{(0,\phi)}^{(\beta)}$
  denotes the random variable counting the number of eigenvalues in the interval $(0,\phi)$ of the circular $\beta$ ensemble (this is essentially the
  same quantity as introduced in the second sentence of the Introduction).
  
  \begin{cor}
  We have
    \begin{align}\label{2.4} 
    S_{N,\beta}(\omega) \Big |_{\omega = 0}  = {N \over 2 \pi} \int_0^\pi \Big ( {\rm Var} \,   {\mathcal N}_{(0,\phi)}^{(\beta)} \Big ) \, d \phi 
     \mathop{\sim}\limits_{N \to \infty} { N \over \beta \pi^2} \Big ( \log {N \over 2} + c_\beta + \cdots \Big ),
    \end{align}
    where terms not shown inside of the brackets of the second line are ${\rm o}(1)$. Here,
    with $\gamma$ denoting Euler's constant, $c_\beta$ takes on the values
     \begin{equation}\label{2.4b} 
     \log 2 + \gamma + 1 - {\pi^2 \over 8}, \quad \log 2 + \gamma + 1, \quad 2 \log 2 + \gamma +1 + {\pi^2 \over 8}
      \end{equation}
    in the cases  $\beta = 1,2,4$ respectively.
    \end{cor}
    
    \begin{proof}
   The random variable ${\mathcal N}_{(0,\phi)}^{(\beta)}$ takes on integer values $l$ in the range $0 \le l \le N$, and moreover
   ${\rm Pr}( {\mathcal N}_{(0,\phi)}^{(\beta)} = l) =  E_{N,\beta}(l;(0, \phi))$. With $\langle \cdot \rangle$ denoting averaging with respect
   to (\ref{2.0}), it follows from this latter formula that
    \begin{equation}\label{2.4c}  
   \langle   {\mathcal N}_{(0,\phi)}^{(\beta)} \rangle  = \sum_{l=0}^N l  E_{N,\beta}(l;(0, \phi)) = {N \phi \over 2 \pi},   \quad
     \langle   ({\mathcal N}_{(0,\phi)}^{(\beta)})^2  \rangle   = \sum_{l=0}^N l^2  E_{N,\beta}(l;(0, \phi)), 
      \end{equation}   
      where the second equality in the  first formula makes use of the fact that the circular $\beta$ ensemble is translationally invariant, and
      so has constant density ${N \over 2 \pi}$. The relevance of both formulas in (\ref{2.4c}) is seen from the small $\omega$ expansion
   \begin{align}\label{2.4d} 
   \mathcal E_{N,\beta}((0, \phi);1 - z) \Big |_{z = e^{i \omega}} & :=  \sum_{l=0}^N e^{i l \omega}  E_{N,\beta}(l;(0, \phi))  \nonumber \\
    & = 1 + i \omega   \langle   {\mathcal N}_{(0,\phi)}^{(\beta)} \rangle  - {\omega^2 \over 2}  \langle   ({\mathcal N}_{(0,\phi)}^{(\beta)})^2  \rangle + {\rm O}(\omega^3).
  \end{align} 
  Substituting (\ref{2.4d}) in (\ref{2.3}) shows that the leading terms inside the big brackets therein is ${\rm O}(\omega^2)$, while the leading term of the
  prefactor for small $\omega$ is proportional to $1/\omega^2$. Straightforward simplification, together with the fact that
    \begin{equation}\label{2.4e}  
   {\rm Var} \,   {\mathcal N}_{(0,\phi)}^{(\beta)} :=      \langle   ({\mathcal N}_{(0,\phi)}^{(\beta)})^2  \rangle  -  \langle   {\mathcal N}_{(0,\phi)}^{(\beta)}  \rangle^2 =
    \langle   ({\mathcal N}_{(0,\phi)}^{(\beta)})^2  \rangle  - \Big (  {N \phi \over 2 \pi}  \Big )^2,
    \end{equation}
    then leads to the equality in (\ref{2.4}).
    
    The large $N$ asymptotic formula in (\ref{2.4}) follows from the equality therein upon substituting the known large $N$ asymptotic
    formula for the integrand \cite[Eq.~(41)]{SLMS21a} (see \cite[\S 2.4 and 2.7]{Fo23} for context, as well as \cite{FF04}, \cite{FL20})
     \begin{equation}\label{2.4f}  
  {\rm Var} \,   {\mathcal N}_{(0,\phi)}^{(\beta)}    \mathop{\sim}\limits_{N \to \infty} { 2 \over \beta \pi^2} \Big ( \log N + \log \Big ( \sin {\phi \over 2} \Big ) + c_\beta \Big ).
    \end{equation}      
    \end{proof}

In fact for $\beta = 1,2$ and 4 it is possible to go beyond the previous large $N$ asymptotics and deduce exact finite $N$ results, 
from which the former follow and permit the easy calculation of further corrections to any order in large $N$.
\begin{prop}
The CUE power spectrum possesses the  value at $ \omega=0 $
\begin{equation}
	S_{N,2}(0) = \frac{N}{2\pi^2}\left[ \gamma + \psi(N) + N\psi^{\prime}(N) \right] ,
\label{CUE_zero-intercept}
\end{equation} 
where $ \psi(x), \psi^{\prime}(x) $ are the digamma function and its derivative \cite[\S 5.2]{DLMF} respectively.
This value for the COE power spectrum is
\begin{equation}
	S_{N,1}(0) = N\left[ -\tfrac{1}{8}+\frac{1}{\pi^2}\left( \gamma + \psi(N) + N\psi^{\prime}(N) + \tfrac{1}{4}\psi^{\prime}(\tfrac{1}{2}(N+1)) \right) \right] .
\label{COE_zero-intercept}
\end{equation} 
The corresponding value for the CSE power spectrum is
\begin{equation}
	S_{N,4}(0) = N\left[ \tfrac{1}{32}+\frac{1}{4\pi^2}\left( \gamma + \psi(2N) + 2N\psi^{\prime}(2N) - \tfrac{1}{4}\psi^{\prime}(N+\tfrac{1}{2}) \right) \right] .
\label{CSE_zero-intercept}
\end{equation} 
\end{prop}
We will derive these results in the following subsections treating the different $ \beta $ cases separately.
    
\subsection{The case $\beta = 2$} \label{S2.2}  
  The eigenvalues of the CUE form a determinantal point process, meaning that the general $k$-point
  correlation function can be expressed as a $k$-dimensional determinant, with each element of the
  form $K_N(x_j,x_l)$ for some kernel function $K_N$ independent of $k$ \cite{Dy62a}.
  A minor generalisation of \cite[Eqns.~(111) and (112)]{Dy62a}, which makes essential use of the
  determinantal point process structure, gives that
     \begin{equation}\label{2.5}  
 \mathcal  E_{N,2}((0, \phi);1 - z) \Big |_{z = e^{i \omega}} = \det \bigg [ \delta_{p,q} - (1-e^{i \omega} ){\sin ((p-q) \phi/2) \over \pi  (p-q) } \bigg ]_{p,q=-N+{1/2},\dots, N - 1/2}.
  \end{equation}
  This can be viewed too as resulting from an application of Andr\'eief's integration formula \cite{Fo19}.
  Note that with the matrix in this expression having entries depending on the difference of the row and column
  labels only, i.e.~it is a Toeplitz matrix, one can equivalently take $p,q=0,\dots,N-1$.
  In the case $p=q$ the scaled sinc function herein is to interpreted as being equal to $\phi  / (2 \pi)$.
  In the computer algebra system Mathematica this is a convenient form for purposes of numerical evaluation
  due to the ability to work in high precision arithmetic; otherwise there would be ill-conditioning due to cancellation of
  fixed precision approximation of the entries for large $N$  \cite{St73}. According to Proposition \ref{P2.1}, (\ref{2.5}) is a
  factor of the integrands in which $\phi$ is integrated over $(0,\pi)$. Relevant then is its graphical form for specific $\omega$ and $N$;
  an example is given in Figure \ref{Fa1} which represents typical behaviour (as $N$ is further increased, the number of oscillations
  increases while their amplitude decreases, and similarly as $\omega$ increases).  
  
\begin{figure*}
\centering
\includegraphics[width=0.7\textwidth]{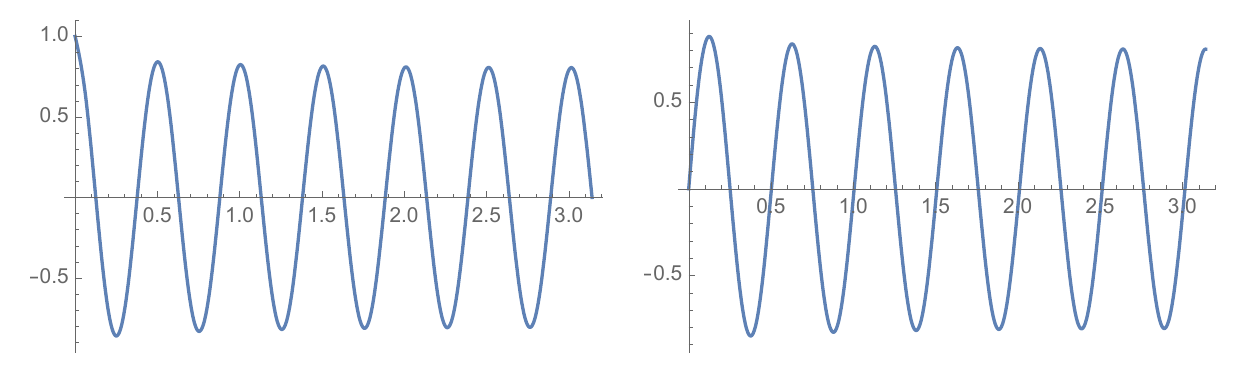}
\caption{Plot of the real part (left figure) and imaginary part (right figure) of (\ref{2.5}) as a function of $\phi \in [0,\pi]$,
with $N=100$ and $\omega = \pi/4$.}
\label{Fa1}
\end{figure*}

  Given the oscillatory nature of the integrand in (\ref{2.4}), as a numerical integration strategy specific to the use of Mathematica, for a given
  $z = e^{i \omega}$, use was made of (\ref{2.5}) to compute $\mathcal E_{N,2}((0, \phi);1 - z)$ to high accuracy at the points $\phi = \pi k/100$ for $k=0,\dots,100$
  and then an interpolation function formed as a continuous (accurate) approximation to this quantity for continuous $\phi$. After this step, the integrations are
  carried out using {\tt NIntegrate}. We display the results of the computation of $S_{N,2}(\omega)$, for $N=20,40$ using such a procedure, in Figure \ref{Fa2}.
  The values at $\omega = 0$ can be checked to be consistent with (\ref{2.4}) for $\beta = 2$ and with (\ref{CUE_zero-intercept}).
  
\begin{figure*}
\centering
\includegraphics[width=0.7\textwidth]{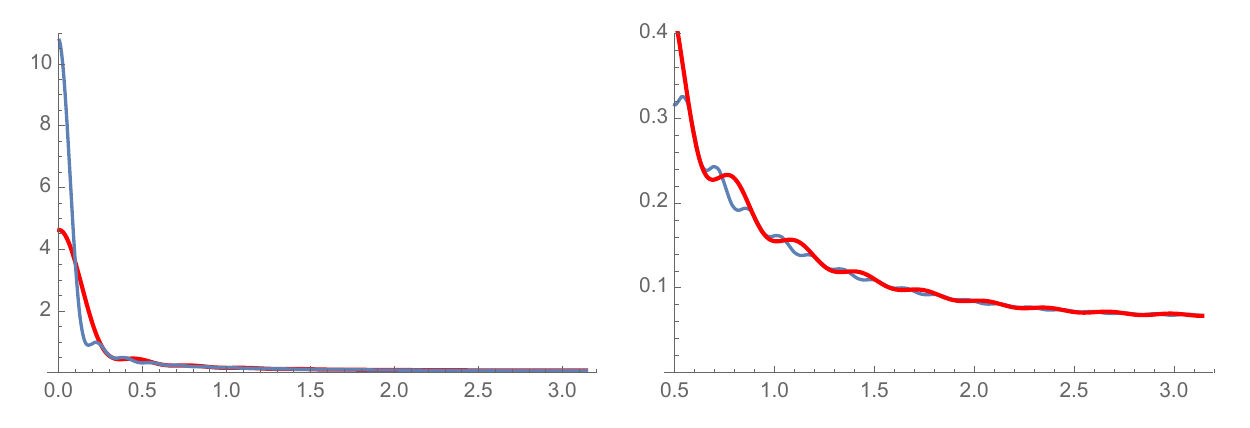}
\caption{[color online] Plot of $S_{N,2}(\omega)$, for $N=20$ (heavy red line) and $N=40$ over the full range of $\omega$, and over a restricted range, the latter  for purposes
of exhibiting evidence of (non-uniform) convergence to a limiting functional form as $N$ increases.}
\label{Fa2}
\end{figure*}

A fundamental result in random matrix theory, with its origins in the work of Gaudin \cite{Ga61}, gives that the matrix determinant form (\ref{2.5}) can
be re-expressed as a Fredholm determinant
   \begin{equation}\label{2.5x}  
 \mathcal E_{N,2}((0, \phi);1 -   e^{i \omega}) = \det \Big ( \mathbb I - (1-e^{i \omega}) \mathcal  K^{(0,\phi)}_N \Big ),
  \end{equation}
  where $\mathcal K_N^{(0,\phi)}$ is the integral operator on $(0,\phi)$ with kernel 
   \begin{equation}\label{2.5a}    
  K_{N,2}(\theta,\theta') = {1 \over 2 \pi} {\sin {1 \over 2} N (\theta - \theta') \over \sin {1 \over 2} (\theta - \theta')},
  \end{equation}
  while $\mathbb I$ is the identity operator.
  In relation to this, one recalls that  the Fredholm determinant can be defined in terms of a product over  the eigenvalues
  $\{ \lambda_j \}_{j=1}^N$ of the finite rank integral operator $\mathcal K_N^{(0,\phi)}$, 
  \begin{equation}\label{2.5b} 
    \mathcal E_{N,2}((0, \phi);1 - e^{i \omega}) =  \prod_{l=1}^N \Big ( 1 - (1 - e^{i \omega}) \lambda_l \Big ).
    \end{equation}
    Since the work of Bornemann \cite{Bo08,Bo10} it has been known that there is an efficient and accurate
    computational scheme --- a Nystr\"om-type quadrature approximation --- for the  evaluation  of Fredholm determinants
    with analytic kernels as in (\ref{2.5a}).
    This then provides an alternative to working with
    the $N \times N$ determinant expression in (\ref{2.5}); for more on this see Section \ref{S3}.
    
  In the scaled large $N$
    limit of (\ref{2.5x}) for which $\phi$ is replaced by $2 \pi s/N$ and the limit $N \to \infty$ then taken, 
    the Fredholm expression on the right hand side remains valid but 
    with $\mathcal K_N^{(0,\phi)}$ therein replaced by the integral $\mathcal K_\infty^{(0,s)}$ with kernel
     \begin{equation}\label{2.5c}   
   K_{\infty,2}(x,y)   = {1 \over \pi} {\sin \pi (x - y) \over x - y} = {\rm sinc} \, \pi (x - y);
  \end{equation} 
  see e.g.~\cite[Ch.~9]{Fo10}.  
  In keeping with foundational work of the Kyoto school relating this limiting Fredholm determinant to the theory of
  integrable systems \cite{JMMS80}, there is a so-called $\tau$-function expression for (\ref{2.5})
  \cite{FW04,Fo10,RK23}
   \begin{equation}\label{2.5d} 
 \mathcal   E_{N,2}((0, \phi);\xi)=  \exp \bigg (- \int_{\cot \phi/2}^\infty u_N(t;\xi) \, {dt \over 1 + t^2} \bigg ).
  \end{equation} 
  Here $u_N(t;\xi)$ is a particular $\sigma$-form Painlev\'e VI transcendent, characterised as satisfying the corresponding PVI $\sigma$-form 
  nonlinear differential  equation
  \begin{equation}\label{2.5e}  
  \Big ( (1 + t^2) u_N'' \Big )^2      + 4 u_N' (u_N - t u_N')^2 + 4 (u_N')^2 (u_N' + N^2) = 0
  \end{equation}
 subject to the $t \to \infty$ boundary condition\footnote{Here we have corrected the reporting in earlier literature.}
   \begin{equation}\label{2.5f}   
   u_N(t;\xi) =  {\xi \over  \pi} N + {\xi^2 \over  \pi^2} N^2 t^{-1} + {\xi^3 N^3 \over \pi^3}  t^{-2} +
   {\rm O}(t^{-3}).
   \end{equation} 
   
First noted in \cite{FW05}, and applied as a practical computational scheme for $S_N(\omega)$ in \cite{ROK20}, integrable systems theory can
also be used to compute the generating function $\mathcal E_{N,2}$ by recurrence. 
For this consider the particular, equivalent to discrete PV equation, 
\begin{equation}\label{2.5g} 
   2 r_n r_{n-1} + 2 \cos \phi = {1 - r_n^2 \over r_n} \Big ( (n+1) r_{n+1} +
   (n-1) r_{n-1} \Big ) -      {1 - r_{n-1}^2 \over r_{n-1}} \Big ( nr_n + (n-2) r_{n-2} \Big ),
\end{equation}
subject to the initial values
\begin{equation}\label{2.5h}  
     r_{-1} = 0, \quad r_0 = 1, \quad r_1 = \xi {\sin \phi \over \pi - \xi \phi}.
\end{equation}
One then has that
\begin{equation}\label{2.6}
    \log \mathcal E_{N+1,2}((0,2 \phi);\xi) = \sum_{j=1}^N ( N + 1 - j) \log (1 - r_j^2) + (N+1) \log \Big (1 - {\xi \over \pi} \phi \Big ).
\end{equation}
Also relevant in this setting are the specialised system of (bi)-orthogonal polynomials $ \Phi_{n}(z=\pm 1) $, which satisfy the recurrence relations
\begin{align}
	\frac{1}{r_{n+1}}\Phi_{n+1}(1) & = \left[ \frac{1}{r_{n+1}}+\frac{1}{r_{n}} \right]\Phi_{n}(1) + \left[ r_{n}-\frac{1}{r_{n}} \right]\Phi_{n-1}(1) ,
\nonumber \\
	\frac{1}{r_{n+1}}\Phi_{n+1}(-1) & = \left[ -\frac{1}{r_{n+1}}+\frac{1}{r_{n}} \right]\Phi_{n}(-1) + \left[ -r_{n}+\frac{1}{r_{n}} \right]\Phi_{n-1}(-1) ,
\label{Phi-_recur}	
\end{align}
with initial conditions $ \Phi_{-1}(\pm 1) = 0 $, $ \Phi_{0}(\pm 1) = 1 $.
These polynomials $ \Phi_{n}(z) $ are orthogonal with respect to the weight $ w(e^{i\theta}) = 1-\xi\chi_{(-\phi,\phi)}(\theta) $ and monic.
 
 Relevant to the results (\ref{CUE_zero-intercept})--(\ref{CSE_zero-intercept}) are expansion of all these quantities
 about $\xi=0$ (equivalently $\omega=0$) for fixed $n,N$.
\begin{lemma}
Let the $l$-th Taylor coefficient of $ r_{n} $ be denoted $ R_{n,l} $, $r_{n} = \delta_{n,0}+R_{n,1}\xi + {\rm O}(\xi^2) $.
Then for $ 0 < \phi < \pi $ we have  
\begin{equation}
	R_{n,1} = \frac{\sin n\phi}{n\pi}, \; n\geq 1,\quad R_{0,1} = 0.
\label{r_Expansion}
\end{equation}
Let the $l$-th Taylor coefficient of $ \Phi_{n}(\pm 1) $ be denoted by $ \alpha_{n,l}, \beta_{n,l} $ respectively, 
$ \Phi_{n}(1) = 1+\alpha_{n,1}\xi + {\rm O}(\xi^2) $ and $ \Phi_{n}(-1) = (-1)^{n}+\beta_{n,1}\xi + {\rm O}(\xi^2) $.
Then for $ 0 < \phi < \pi $ we have 
\begin{equation}
	\alpha_{n,1} = \frac{1}{\pi} \sum_{j=1}^{n}\frac{\sin j\phi}{j} ,
\quad
	\beta_{n,1} = \frac{(-1)^{n}}{\pi} \sum_{j=1}^{n}(-1)^{j}\frac{\sin j\phi}{j} .
\label{1st-order_BOPS}
\end{equation}
The first two terms in the expansion of $ {\mathcal{E}}_{N,2}((0,2\phi);\xi) $ are given by
\begin{equation}
	\mathcal{E}_{N,2}((0,2\phi);\xi) = 1 - \frac{N\phi}{\pi}\xi
	+ \left[
		\frac{N(N-1)}{2\pi^2}\phi^2
		-\frac{1}{\pi^2}\sum_{j=1}^{N-1}(N-j)\frac{\sin^2(j\phi)}{j^2}
	\right] \xi^2	
	+ {\rm O}(\xi^3) .
\label{ECUEGap_Expansion}
\end{equation}
\end{lemma}

According to (\ref{2.4d}) and (\ref{2.4e}), ${\rm Var} \, \mathcal N_{(0,\phi)}^{(\beta)} |_{\beta = 2}$ can be read off from
 \eqref{ECUEGap_Expansion}. Integrating according to the first equality in (\ref{2.4}) gives  \eqref{CUE_zero-intercept}.

\begin{remark}\label{R1a}
   The determinant expression (\ref{2.5}) implies that the coefficients of $z^k$ $(k=0,\dots,N)$ in the generating function $\mathcal E_{N,2}((0,\phi);1-z)$
   are a linear combination of monomials $\phi^p$ $(p=0,\dots,N$) times complex exponentials from $\{e^{iq\phi/2}\}_{q=-N^*}^{N^*}$ for some $N^*>0$.
   This structure allows the integrals in (\ref{2.2}) to be evaluated in closed form, and thus an exact evaluation of $S_{N,2}(\omega)$ carried out, at least
   for small $N$. For example, we can calculate
\begin{multline}\label{2.6h}
    S_{N,2}(\omega) \Big |_{N=3} = { 3(-81 - 72 \pi^2 + 16 \pi^4) \over 128 \pi^4} + { 3(-105 - 60 \pi^2 + 16 \pi^4) \over 160 \pi^4} \cos \omega \\+
    { 3(825 - 120 \pi^2 + 16 \pi^4) \over 640 \pi^4} \cos 2 \omega.
\end{multline}   
\end{remark}
   
\subsection{The case $\beta = 1$}\label{S2.3}
   In distinction to the CUE, the eigenvalues of the COE form a Pfaffian point process. While still structured, this implies a more
   complicated functional form for the analogue of (\ref{2.5}), due to this formalism giving a double integral formula for the
   elements of the Pfaffian; see e.g.~\cite[Prop.~6.3.4]{Fo10}. The recent study \cite{FK24} relating to computing conditioned gap
   probabilities for the Gaussian orthogonal ensemble --- this too being a Pfaffian point process --- shows this extra difficulty not to
   be insurmountable if one's aim is computation of the generating function $\mathcal E_{N,1}$ for fixed $N$. While this is part of our aim, we want
   to also be able to compute the large $N$ limit, where it is (analogues of) the forms (\ref{2.5a}) and/ or  (\ref{2.5d}) which are needed.
   Fortunately, due to certain inter-relations between conditioned
   gap probabilities in the COE, and those for the  determinantal point process corresponding to the
   eigenvalues of the orthogonal groups $O^+(N+1)$ and $O^-(N+1)$
    \cite{FR01,BF15,BFM17}, such functional forms are available.
    
\begin{prop}\label{P2.2} (\cite[Eq.~(5.10)]{BFM17})     
Using superscripts rather that subscripts to indicate the random matrix ensemble that the conditioned gap probability
generating function is referring to, we have that
  \begin{equation}\label{2.7} 
\mathcal  E^{{\rm COE}_N}((0,\phi);\xi) = \frac{1 - \xi}{ 2 - \xi}\mathcal E^{O^\nu(N+1)}((0,\phi/2);\xi(2 - \xi)) + \frac{1}{2 - \xi}\mathcal E^{O^{-\nu}(N+1)}((0,\phi/2);\xi(2 - \xi)) ,
  \end{equation} 
  where $\nu = (-)^N$.
  \end{prop}
    
  The joint eigenvalue probability density function\footnote{This is with respect to the eigenvalues with angles strictly between $0$ and $\pi$. There can
  also be fixed eigenvalues at $\pm 1$, and furthermore the eigenvalues come in complex conjugate pairs.}
    for the orthogonal groups $O^+(N)$ and $O^-(N)$ depends on the parity of $N$; see e.g.~\cite[\S 2.6]{Fo10}.
  Specifically, these functional forms, for $O^-(2N+2), O^+(2N), O^-(2N+1), O^+(2N+1)$, are proportional to
  \begin{equation}\label{2.7a}
  \prod_{l=1}^N (1 + \cos \theta_l)^{\lambda_1}    (1 - \cos \theta_l)^{\lambda_2} \prod_{1 \le j < k \le N}(\cos \theta_k - \cos \theta_j)^2, \quad 0 < \theta_l < \pi,
    \end{equation} 
    for $(\lambda_1, \lambda_2) = (1,1),(0,0),(1,0),(0,1)$ respectively.
   Consequently, there are four determinant formulas analogous to (\ref{2.5}) of relevance to (\ref{2.7}). These can be deduced from
   (easy to derive --- see \cite[Exercises 5.5 q.6]{Fo10}) ``Hankel $+$ Toeplitz'' formulas for averages over the orthogonal group eigenvalue probability density function.
   
   \begin{lemma}\label{L1} (\cite{BR01+})
   Suppose $a(e^{i \theta}) = a(e^{-i \theta})$ and put $a_j = {1 \over 2 \pi} \int_{-\pi}^\pi a(e^{i \theta}) e^{-i j \theta} \, d \theta$.   We have
   \begin{eqnarray*}
 \det [ a_{j-k} + a_{j+k-1} ]_{j,k=1,\dots,N}
 & =  \Big \langle \prod_{j=1}^{N}a(e^{i \theta_j})
\Big \rangle_{{O^-(2N+1)}},  \\
\det [ a_{j-k} - a_{j+k-1} ]_{j,k=1,\dots,N}
 & =  \Big \langle \prod_{j=1}^{N}a(e^{i \theta_j})
\Big \rangle_{{O^+(2N+1)}}, \nonumber \\
{1 \over 2} \det [ a_{j-k} + a_{j+k-2} ]_{j,k=1,\dots,N}  & = 
\Big \langle \prod_{j=1}^{N} a(e^{i \theta_j})
\Big \rangle_{{O^+(2N)}}, \\
\det [ a_{j-k} - a_{j+k} ]_{j,k=1,\dots,N}  & = 
\Big \langle \prod_{j=1}^{N} a(e^{i \theta_j})
\Big \rangle_{{O^-(2N+2)}}.
\end{eqnarray*}
   \end{lemma}
   
   \begin{cor}\label{C2}
   Let $ \mathcal E^{{\rm ME}(N)}((0,\phi);\xi)$ denote the
   generating function (\ref{2.1}) for a matrix ensemble  ME$(N)$.
   With ME$(N)$ one of the particular matrix ensembles $O^-(2N+1), O^+(2N+1),O^+(2N),
   O^-(2N+2)$, we have that $  \mathcal E^{{\rm ME}(N)}((0,\phi);\xi)$ is given by the determinants on the left hand side of listing in Lemma \ref{L1},
   respectively, with
     \begin{equation}\label{2.7b}
      a_j = \delta_{j,0} -  \xi {\sin j  \phi  \over \pi j},
   \end{equation} 
 where here for $j=0$ the function $ {\sin j \phi  \over \pi j}$ is to be replaced by ${\phi \over \pi}$.
     
   \end{cor}
   
  \begin{proof}
  A well known general formula relates the generating function (\ref{2.1}) to an average over the eigenvalue probability
  density function according to
   \begin{equation}\label{2.8} 
  \mathcal E^{{\rm ME}(N)}((0,\phi);\xi) = \Big \langle \prod_{l=1}^N (1 - \xi \mathbbm 1_{\theta_l \in  (0,\phi)}) \Big \rangle_{{\rm ME}(N)};
 \end{equation} 
 see e.g.~\cite[Prop.~8.1.2]{Fo10}.
 Hence Lemma \ref{L1} applies with $a(e^{i \theta}) = 1 -  \xi \mathbbm 1_{\theta \in  (-\phi,\phi)}$, or equivalently with $a_j$ as in
 (\ref{2.7b}).
  \end{proof}  
  
Claeys et al \cite{CGMY22} have found relations for all four orthogonal integrals in terms a single unitary integral, i.e. a $\beta=2$ average, 
plus auxiliary quantities that we have already introduced. In particular the authors have derived
\begin{align}
  	\mathcal{E}^{O^{+}(2N)}((0,\phi);\xi) 
  	& = \left( \frac{\mathcal{E}_{2N,2}((-\phi,\phi);\xi)}{-\Phi_{2N-1}(1)\Phi_{2N-1}(-1)} \right)^{1/2} ,
\label{UEtoEOC:a}\\
  	\mathcal{E}^{O^{-}(2N+2)}((0,\phi);\xi) 
  	& = \left( \Phi_{2N}(1)\Phi_{2N}(-1)\mathcal{E}_{2N,2}((-\phi,\phi);\xi) \right)^{1/2} ,
\label{UEtoEOC:b}\\
  	\mathcal{E}^{O^{\pm}(2N+1)}((0,\phi);\xi) 
  	& = \left( \frac{\Phi_{2N}(\pm 1)}{\Phi_{2N}(\mp 1)}\mathcal{E}_{2N,n}((-\phi,\phi);\xi) \right)^{1/2} ,
\label{UEtoEOC:c}
\end{align}
in Prop. 1.1 and Eq. (1.10) of this work.
These relations permit the derivation of the COE analogues to \eqref{ECUEGap_Expansion}.
\begin{cor}\label{C2.3}
The even case of the COE generating function $ \mathcal{E}^{{\rm COE}_{2N}}((0,\phi);\xi) $ is found to possess the leading small $\xi$ terms
\begin{multline}
	\mathcal{E}^{{\rm COE}_{2N}}((0,\phi);\xi) = 1 - \frac{N\phi}{\pi}\xi
	+ \Bigg[
		\frac{N(N-1)}{2\pi^2}\phi^2
		+\frac{N}{2\pi}\phi
\\
		-\frac{2}{\pi^2}\sum_{j=1}^{2N-1}(2N-j)\frac{\sin^2(\tfrac{1}{2}j\phi)}{j^2}
		+\frac{1}{2\pi^2}\left( \sum_{k=1}^{N}\frac{\sin(k-\tfrac{1}{2})\phi}{k-\tfrac{1}{2}} \right)^2
		-\frac{1}{2\pi} \sum_{k=1}^{N}\frac{\sin(k-\tfrac{1}{2})\phi}{k-\tfrac{1}{2}}
	  \Bigg] \xi^2	
	+ {\rm O}(\xi^3) .
\label{evenECOEGap_Expansion}
\end{multline}
The analogous expansion of the odd case COE generating function $ \mathcal{E}^{{\rm COE}_{2N+1}}((0,\phi);\xi) $ is 
\begin{multline}
	\mathcal{E}^{{\rm COE}_{2N+1}}((0,\phi);\xi) = 1 - \frac{(2N+1)\phi}{2\pi}\xi
	+ \Bigg[
		\frac{N^2}{2\pi^2}\phi^2
		+\frac{N}{2\pi}\phi
\\
		-\frac{2}{\pi^2}\sum_{j=1}^{2N+1}(2N+1-j)\frac{\sin^2(\tfrac{1}{2}j\phi)}{j^2}
		+\frac{1}{2\pi^2}\left( \sum_{k=1}^{N}\frac{\sin k\phi}{k} \right)^2
		+\frac{\phi-\pi}{2\pi^2} \sum_{k=1}^{N}\frac{\sin k\phi}{k}
	  \Bigg] \xi^2	
	+ {\rm O}(\xi^3) .
\label{oddECOEGap_Expansion}
\end{multline} 
\end{cor}
\begin{proof}
In the even case of the COE generating function  the leading terms are found using \eqref{ECUEGap_Expansion}, 
along with \eqref{2.7} and both cases of \eqref{UEtoEOC:c} and involves additional contributions from \eqref{1st-order_BOPS}.
For the odd case  the leading terms are found using \eqref{ECUEGap_Expansion}, 
along with \eqref{2.7} and both \eqref{UEtoEOC:a} and \eqref{UEtoEOC:b} and also involves additional contributions from \eqref{1st-order_BOPS}.
In contrast to the even case we find the presence of second order Taylor coefficients $\alpha_{n,2}$ and $\beta_{n,2}$ in the first two orders of the odd generating function. 
However they only appear in the following combination $\alpha_{2N+1,2}-\alpha_{2N,2}-\beta_{2N+1,2}-\beta_{2N,2}$, and due to the identities
\begin{equation}
	\alpha_{n+1,2}-\alpha_{n,2} = R_{n+1,2} + R_{n+1,1}\sum_{j=1}^{n}R_{j,1} ,
\qquad
	\beta_{n+1,2}+\beta_{n,2} = R_{n+1,2} + R_{n+1,1}\sum_{j=1}^{n}(-1)^jR_{j,1} ,
\label{key}
\end{equation}
all the second order terms cancel out of that particular combination.
\end{proof}

The results of Corollary \ref{C2.3} imply two distinct formulas for  ${\rm Var} \, \mathcal N_{(0,\phi)}^{(\beta)} |_{\beta = 1}$,
depending on the parity of $N$.
Integration according to the first equality in (\ref{2.4})
leads to the same result \eqref{COE_zero-intercept} in both cases

  For small $N$, use of Corollary \ref{C2} in (\ref{2.7}) allows the integrals in (\ref{2.2}) to be evaluated in closed form, 
  which when substituted in (\ref{2.4}) allows for a corresponding closed form evaluation of
  $S_{N,1}(\omega)$. This is analogous to the situation for  $S_{N,2}(\omega)$ for small $N$ (recall Remark \ref{R1a}).
  Specifically, in the case $N=3$ we calculate
  \begin{equation}\label{2.9}  
  S_{N,1}(\omega) = {(-225 - 60 \pi^2 + 14 \pi^4) \over 24 \pi^4} + {(\pi^2 - 6)  \over 2 \pi^2} \cos \omega + {(225 - 48 \pi^2 + 4 \pi^4) \over 24 \pi^4} \cos 2 \omega.
   \end{equation}  
   It is furthermore the case that the use of Corollary \ref{C2} in (\ref{2.7}), together with (\ref{2.2}), allows for a numerical
   evaluation of $S_{N,1}(\omega) $. This  can be carried out along the same lines as that detailed for the numerical
   evaluation of  $S_{N,2}(\omega)$ in the paragraph after the one containing (\ref{2.5}).
  The results for $N=20,40$ are displayed in Figure \ref{Fa3}, with the result (\ref{COE_zero-intercept}) relating to the value at
  $\omega = 0$ providing a consistency check. Also, since $N$ is even in our examples, only the first two identities in
  Lemma \ref{L1} are required.
  
\begin{figure*}
\centering
\includegraphics[width=0.8\textwidth]{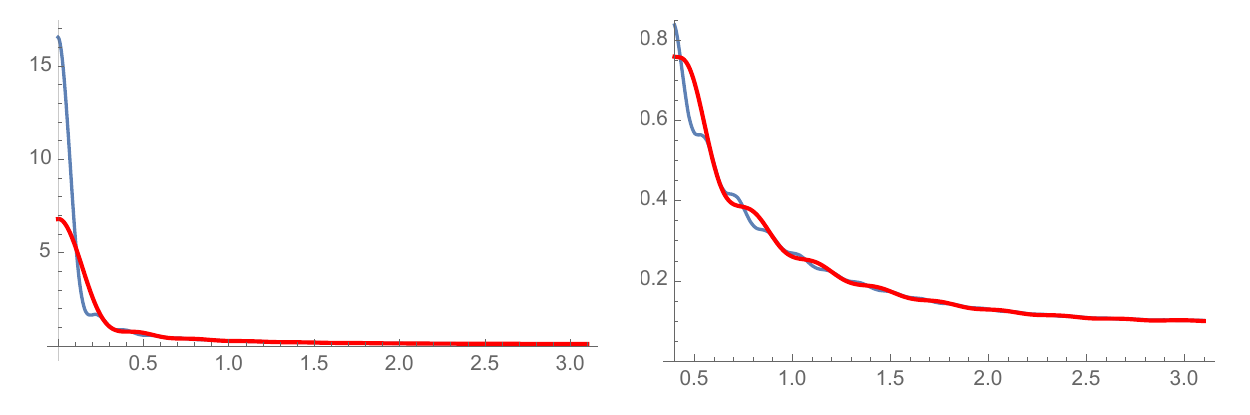}
\caption{[color online] Plot of $S_{N,1}(\omega)$, for $N=20$ (heavy red line) and $N=40$ over the full range of $\omega$, and over a restricted range. Note that the latter  
exhibits evidence of (non-uniform) convergence to a limiting functional form as $N$ increases.}
\label{Fa3}
\end{figure*}
  
  Analogues of (\ref{2.5x}) for the matrix averages in Lemma \ref{L1}, which as for $\beta = 2$ provide an alternative computational
  scheme for the evaluation of $S_{N,1}(\omega)$. For brevity of presentation, we will restrict attention to the case of $N$ even and
  thus (as noted at the end of the previous paragraph) only the averages with respect to $O^\pm(2N+1)$ are required. The key point is that
  the eigenvalues form a determinantal point process, here with kernel (see e.g.~\cite[Prop. 5.5.3]{Fo10})
 \begin{equation}\label{K1}
 K_N^{O^\pm(2N+1)}(\theta, \theta') = {1 \over 2 \pi}  \bigg ( {\sin N (\theta - \theta') \over \sin((\theta - \theta')/2)} \mp
 {\sin N (\theta + \theta') \over \sin((\theta + \theta')/2)} \bigg )
    \end{equation} 
 (cf.~(\ref{2.5a})). The analogue of (\ref{2.5x}) is then
  \begin{equation}\label{K1a} 
  \mathcal  E^{O^\pm(2N+1)} ((0, \phi);\xi) = \det \Big ( \mathbb I - \xi \mathcal  K^{O^\pm(2N+1) , (0,\phi)}_N \Big ),
 \end{equation}
 where $\mathcal  K^{O^\pm(2N+1) , (0,\phi)}_N$ is the integral operator on $(0,\phi)$ with kernel (\ref{K1}),
  for which the numerical methods of \cite{Bo08,Bo10} apply.
  
  It is furthermore the case that $ E^{O^\pm(2N+1)} ((0, \phi);\xi)$ permits a $\tau$-function expression resulting from 
  Okamoto's Hamiltonian theory of the Painlev\'e VI equation \cite{FW04}. Explicitly \cite[Eq.~(5.26)--(5.29)]{BFM17}
  \begin{equation}\label{Sa1a}
  \mathcal   E^{{\rm O}^\pm(2N+1)}((0,\phi);\xi) =  \exp \int_0^{\sin^2 \phi/2} {u^\pm(t;\xi) \over t ( t - 1)} \, dt,
 \end{equation}
 where $u=u^\pm$ satisfies the particular $\sigma$PVI equation
 \begin{equation}\label{Sa3}  
 (t(1-t)u'')^2 + (u' - N^2)(2 u + (1 - 2t)u')^2 - (u')^2 \Big (u' - N^2 + \tfrac{1}{4} \Big ) = 0,
 \end{equation}
 subject to the boundary condition
 \begin{equation}\label{Sa4} 
 u^\pm(t;\xi) \mathop{\sim}\limits_{t \to 0^+} \left \{ \begin{array}{ll} \displaystyle {8 N(N^2-1/4) \over 3 \pi} \xi t^{3/2} ,  & O^+(2N+1)\\[.2cm]
  \displaystyle {2 N \over \pi} \xi t^{1/2} (1 + {\rm O}(t)), &   O^-(2N+1). \end{array} \right.
 \end{equation}
 
We remark that the even cases $ E^{O^+(2N)}((0,\phi);\xi)$ and $ E^{O^-(2N+2)}((0,\phi);\xi)$ are characterised by another pair of Painlev\'e VI $\tau$-functions,
satisfying a very similar $\sigma$-form and distinguished only by their boundary conditions. These follow from (\ref{2.7a}) with the change of variables
$(1 - \cos \theta_l)/2= x_l$, and the results of \cite[\S 3.1]{FW04} or \cite[\S 8.3.1]{Fo10}.
 
 \subsection{The case $\beta = 4$}\label{S2.4}
 As for the COE, the eigenvalues of the CSE form a Pfaffian point process. However the conditioned spacing probabilities of the
 COE and CSE are inter-related due to an identity relating the joint eigenvalue PDF for the COE with $2N$ eigenvalues, integrated over
 alternative eigenvalues, to the eigenvalue PDF of the CSE \cite{MD63}. On the other hand, we already know that the
 conditioned spacing probabilities of the
 COE are related to those for the orthogonal group, suggesting the same for the CSE. In fact the simple relation \cite[Eq.~(8.158) ]{Fo10},
 \cite[Eq.~(5.16)]{BFM17}
 \begin{equation}\label{2.7S} 
  \mathcal   E^{{\rm CSE}_N}((0,\phi);\xi) = {1 \over 2} \Big (   \mathcal   E^{O^+(2N+1)}((0,\phi/2);\xi)   +   \mathcal  E^{O^{-}(2N+1)}((0,\phi/2);\xi) \Big )
  \end{equation} 
 holds true. Combined with Lemma \ref{L1}, Corollary \ref{C2} and  (\ref{2.2}), this offers both a scheme for the exact evaluation of
 $S_{N,4}(\omega)$ (at least for small $N$), as well as an efficient numerical approach the computation of $S_{N,4}(\omega)$.

\begin{cor}
The leading terms in the $\xi$-expansion of the CSE generating function are given by
\begin{multline}
	\mathcal{E}^{{\rm CSE}_N}((0,\phi);\xi) = 1 - \frac{N\phi}{2\pi}\xi
\\
	+ \left[
		\frac{N(N-1)}{8\pi^2}\phi^2
		-\frac{1}{\pi^2}\sum_{j=1}^{2N-1}(N-\tfrac{1}{2}j)\frac{\sin^2(\tfrac{1}{2}j\phi)}{j^2}
		+\frac{1}{8\pi^2}\left( \sum_{k=1}^{N}\frac{\sin(k-\tfrac{1}{2})\phi}{k-\tfrac{1}{2}} \right)^2
	\right] \xi^2
\\	
	+ {\rm O}(\xi^3) ,
\label{ECSEGap_Expansion}
\end{multline}
\end{cor} 
\begin{proof}
These leading terms are found using \eqref{ECUEGap_Expansion}, along with \eqref{2.7S}, and both cases of \eqref{UEtoEOC:c}. 
This also involves additional contributions from \eqref{1st-order_BOPS}.
Interestingly only the first order contributions to the bi-orthogonal polynomials are required as the second order terms cancel out. 
\end{proof}

From  \eqref{ECSEGap_Expansion} we read off the evaluation of 
 ${\rm Var} \, \mathcal N_{(0,\phi)}^{(\beta)} |_{\beta = 4}$. 
The CSE intercept value \eqref{CSE_zero-intercept}  follows upon the integration according to the first equality in (\ref{2.4}). 

The formula (\ref{2.7S}) can be used in conjunction with Corollary \ref{C2} to obtain explicit evaluations of  $S_{N,4}(\omega)$ for small $N$, and
numerical tabulations for moderate $N$. In relation to the former, for $N = 3$ as an example, we calculate
  \begin{multline}\label{2.7T} 
  S_{N,4}(\omega) \Big |_{N=3} = { (-148225 - 110880 \pi^2 + 20736 \pi^4) \over 55296 \pi^4} + { (-145145 - 92400 \pi^2 + 20736 \pi^4 )\over 69120 \pi^4} \cos \omega \\+
    { (1321705 - 184800 \pi^2 + 20736 \pi^4) \over 276480 \pi^4} \cos 2 \omega.
    \end{multline} 
    In relation to the latter, we can make use of code already used for the numerical computation of
    $S_{N,1}$ due to the appearance of $ E^{O^\pm(2N+1)}((0,\phi/2);z) $  in both. The results for $N=20,40$ are displayed in Figure \ref{Fa4}, with the result (\ref{2.4}) relating to the value at
  at $\omega = 0$ again providing a consistency check.

         \begin{figure*}
\centering
\includegraphics[width=0.8\textwidth]{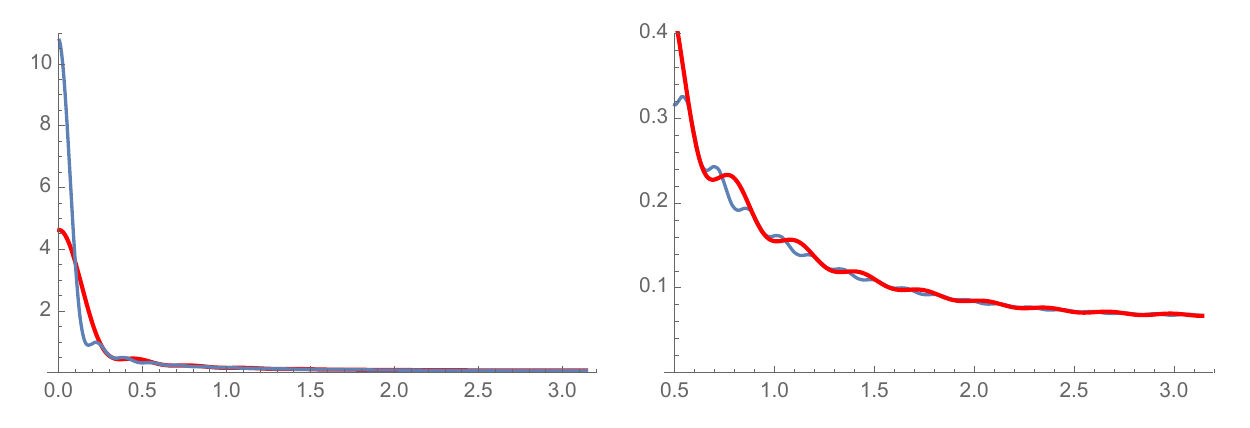}
\caption{[color online] Plot of $S_{N,4}(\omega)$, for $N=20$ (heavy red line) and $N=40$ over the full range of $\omega$, and over a restricted range. We again note that the latter  
exhibits evidence of (non-uniform) convergence to a limiting functional form as $N$ increases.}
\label{Fa4}
\end{figure*}

We remark that use of (\ref{K1a}) in (\ref{2.7S}) provides an expression for $S_{N,4}$ in terms of Fredholm determinants, while use of (\ref{Sa1a}) provides
for an expression in terms of two PVI $\tau$-functions.

\section{The $N \to \infty$ limit of $S_{N,\beta}(\omega)$}\label{S3}
\subsection{Analytic aspects}
We see from Proposition \ref{P2.1} that the key to calculating the $N \to \infty$ limit of $S_{N,\beta}(\omega)$ is to quantify the corresponding asymptotic behaviour of 
the generating function $ \mathcal E_{N,\beta}((0,\phi);1 - e^{i \omega})$. For this one notes that it follows from the definition (\ref{2.1}), and the definition of $\mathcal N_{(0,\phi)}^{(\beta)}$ as
introduced below Proposition \ref{P2.1}, that
 \begin{equation}\label{3.1}
 \mathcal E_{N,\beta}((0,\phi);1 - e^{i \omega}) = \Big \langle  e^{i \omega \mathcal N_{(0,\phi)}^{(\beta)} } \Big \rangle,
\end{equation} 
where the average $\langle \cdot \rangle$ is with respect to the circular $\beta$ ensemble eigenvalue probability density function (\ref{2.0}).
The significance of this identification is the known asymptotic expansion \cite{FL20} (see also \cite[Eq.~(49)]{SLMS21a})\footnote{Beyond the case
$\beta = 2$, this is based on a yet to be proved conjecture from \cite{FF04}. In the case $\beta = 2$, the expansion has
been extended one term further; see \cite[\S 2.3]{KKMST09}.}
 \begin{equation}\label{3.2}
 \log  \Big \langle  e^{i \omega( \mathcal N_{(0,\phi)}^{(\beta)}  - N \phi/(2 \pi) )} \Big \rangle = 2 \Big (  t^2 \log \Big (
 2 N \sin {\phi \over 2}  \Big ) + \log (A_\beta(t)A_\beta(-t))   + \cdots \Big ) \Big |_{t = i \omega/ \sqrt{2 \pi^2 \beta}}
\end{equation} 
Here, for $\beta = 1,2,4$ the constant (i.e.~independent of  $N$ and $\phi$) $A_\beta(t)$ is given in terms of the Barnes $G$ function\footnote{This
special function is related to the gamma function via the functional relation $G(z+1) = \Gamma(z) G(z)$.} by
 \begin{equation}\label{3.2z}
 2^{-t^2/2} G \Big ( 1 + {it \over \sqrt{2 }} \Big ) {G \Big ( 3/2 + {it \over \sqrt{2 }} \Big )  \over G(3/2)}, \quad G \Big ( 1 + it \Big ), \quad
 G \Big ( 1 +{ it  \over \sqrt{2}}\Big ) {G \Big ( 1/2 + {it \over \sqrt{2 }} \Big ) \over G(1/2)}
 \end{equation} 
 respectively, and in fact can be expressed in terms of this special function for all $\beta$ rational, 
while the terms not shown go to zero as
$N \to \infty$. 

For fixed $\phi$ one observes that (\ref{3.2}) is consistent with the Gaussian fluctuation theorem of Killip \cite[Eq.~(1.5)]{Ki08}, which asserts
 \begin{equation}\label{3.2a}
 \lim_{N \to \infty} \Big ( {\pi^2 \beta \over 2 \log N} \Big )^{1/2} \Big (  \mathcal N_{(0,\phi)}^{(\beta)}  - N \phi/(2 \pi) \Big ) \mathop{=}^{\rm d} {\rm N}[0,1].
 \end{equation}
 Important for us is the modification of (\ref{3.2}) to the setting that $\phi$ simultaneously goes to zero as $N$ goes to infinity, but slow enough that
$\phi N \to \infty$. Formally this is obtained by simply replacing $ \sin {\phi \over 2}$ by its leading small $\phi$ form ${\phi \over 2}$, so obtaining
\begin{equation}\label{3.2X}
  \Big \langle  e^{i \omega \mathcal N_{(0,\phi)}^{(\beta)} }  \Big \rangle  \mathop{\sim}_{\phi N \to \infty} 
  e^{i \omega N \phi/(2 \pi) - {\omega^2 \over 2} \sigma_{\beta,N} + {\rm O}(1)}, \quad \sigma_{\beta,N}^2 = {2 \over \beta \pi^2} \log \phi N.
\end{equation} 
The validity of (\ref{3.2X}) for $\beta = 2$ follows from the underlying determinantal structure. Thus (up to the bound on the error term 
which itself follows from the more refined analysis of \cite{CK15}), the determinantal structure implies all that is required for the validity of such
a Gaussian fluctuation formula is the limiting divergence of the variance $ \sigma_{\beta,N}$
 \cite{CL95}, \cite{So00}. Moreover, known superposition and interlacing inter-relations linking the COE, CUE, CSE 
(see the recent review \cite{Fo24}) were shown in  \cite{CL95}, \cite{OR10} to lift the validity of (\ref{3.2X}) from $\beta = 2$ to the cases $\beta = 1$ and 4 as well.

Also important for us is that for $\beta = 1,2$ and 4, by making use of the Fredholm determinant formulas implied by  (\ref{2.5x}) and (\ref{K1a}), it is a straightforward matter to
establish the uniform convergence 
\begin{equation}\label{3.2Y}
\mathcal E_{N,\beta}((0,2 \pi s/N);1 - e^{i \omega}) \to   \mathcal E_{\infty, \beta} ((0,s); 1 - e^{i \omega}), \quad 0 \le s < s_0.
\end{equation} 
Here the $ \mathcal E_{\infty, \beta}$ are bounded, explicit functional forms (given in Proposition \ref{P3.2} below),
which moreover exhibit asymptotics compatible with (\ref{3.2X}) for large $s$ (given in
 (\ref{3.2f}) below).

 As a consequence of (\ref{3.2X}) and (\ref{3.2Y}) (proved for $\beta = 1,2,4$ but expected to remain true for general $\beta > 0$), the large $N$ asymptotics of the integrals (\ref{2.2}) can be deduced, and an expression obtained for 
 the $N \to \infty$ limit of $S_{N,\beta}(\omega)$. (See  \cite[\S 4.1]{RK23} for independent working in the case $\beta = 2$.)
 
 \begin{prop}
  We have
 \begin{equation}\label{3.2b}
 \lim_{N \to \infty} I_{N,0}^{(\beta)}(z) = {1 \over 2 \pi} \int_0^\infty  \mathcal E_{\infty,\beta}((0,s);1 - z) \, d s, \quad  \lim_{N \to \infty} I_{N,1}^{(\beta)}(z) = 0,
 \end{equation} 
 where $  \mathcal E_{\infty,\beta}((0,s);1 - z)$ is the conditioned gap probability generating function for the $N \to \infty$ limit of the Dyson circular ensemble  \cite{KN04},
 scaled so that the density is unity. Consequently
  \begin{equation}\label{3.2c}
  S_{\infty,\beta}(\omega) := \lim_{N \to \infty} S_{N,\beta}(\omega) = {1 \over  2 \sin^2 {\omega \over 2}} {\rm Re}
  \int_0^\infty  \mathcal E_{\infty, \beta} ((0,s); 1 - e^{i \omega}) \, ds.
   \end{equation} 
\end{prop} 

\begin{proof}
In (\ref{2.2}) we break the integral up into the regions $\phi \in (0,\Omega(N)/N)$ and $\phi \in (\Omega(N)/N, \pi)$, where $1 \ll \Omega(N) \ll N$.
Use of the asymptotic formula (\ref{3.2X}) shows that the latter region does not contribute. 
 For the integral over $ (0,\Omega(N)/N)$, one changes variables $\phi = 2 \pi  s/N$ and makes use uniform convergence  (\ref{3.2Y}) with tail
 bound compatible with (\ref{3.2X}). From this we conclude (\ref{3.2b}). The formula (\ref{3.2c}) now follows  by applying (\ref{3.2b}) to (\ref{2.3}).
\end{proof}

\begin{remark} 
Let $\mathcal N_{(0,s)}^{\infty,(\beta)}$ denote the random variable corresponding to the counting function
for the number of eigenvalues in the interval $(0,s)$ of the infinite circular $\beta$ ensemble, scaled 
to have unit density. A result of Costin and Lebowitz \cite{CL95} has established the central limit theorem 
\begin{equation}\label{3.2d}
 \lim_{s \to \infty}{1 \over \sigma_s} \Big (  \mathcal N_{(0,s)}^{\infty,(\beta)}  - s \Big ) \mathop{=}^{\rm d} {\rm N}[0,1], \quad \sigma_s^2
:= 2  (\log s)/ (\pi^2 \beta)
 \end{equation}
 (cf.~(\ref{3.2a})) valid  for $\beta = 1,2$ and 4 at least. Consequently, we have  the large $s$ asymptotic expansion
 \begin{equation}\label{3.2f} 
 \mathcal E_{\infty, \beta} ((0,s); 1 - e^{i \omega}) = \Big \langle  e^{i \omega \mathcal N_{(0,s)}^{\infty,(\beta)} } \Big \rangle \mathop{\sim}_{s \to \infty}
 e^{i  \omega s- (\omega^2/2) \sigma_s^2 + {\rm O}(1)},
  \end{equation}
   (cf.~(\ref{3.2X}))
  where here in the equality the average is with respect to the unit density infinite circular $\beta$ 
 ensemble\footnote{Formally writing $\phi = 2 \pi s /N$ in (\ref{3.2}), where $1 \ll s \ll N$ reclaims (\ref{3.2f}). This procedure furthermore predicts that the
  O$(1)$ term is equal to
  $$
  2 t^2 \log 2 \pi + 2 \log(A_\beta(t) A_\beta(-t)) \Big |_{t = i \omega / \sqrt{2 \pi^2 \beta}},
  $$
  where $A_\beta(t)$ is given explicitly in (\ref{3.2z}) for $\beta=1,2$ and 4.}.
   The convergence of the integral in (\ref{3.2c})
  at infinity follows from this large $s$ form, which itself is oscillatory and decaying. 
  \end{remark}
  
  There is particular interest in the small $\omega$ form of $S_{\infty,\beta}(\omega)$, as in applications to quantum
  chaos this distinguishes chaotic from integrable spectra \cite{RGMRF02,ROK17} .
  In keeping with Fourier transform theory \cite{Li58}, we expect that the small $\omega$ singular
  behaviour of $S_N(\omega)$ is determined by the large distance asymptotics of $ \mathcal E_{\infty, \beta} ((0,s); 1 - e^{i \omega})$,
  this being the integrand in (\ref{3.2c})\footnote{For a recent work in random matrix theory where this aspect of Fourier transform theory is
  put to application, see \cite{FS23}.}.
Assuming this, (\ref{3.2f}) and (\ref{3.2c}) then give
 \begin{multline}\label{3.2g} 
 S_{\infty,\beta}(\omega)  \mathop{\sim}_{\omega \to 0}^{\cdot}{2 \over \omega^2}  \int_0^\infty s^{-\omega^2/(\pi^2 \beta)} \cos \omega s \, ds ={2 \over \omega^2}  |\omega|^{-1+\omega^2/(\pi^2\beta)} \sin(\omega^2/(2 \pi \beta))
 \Gamma(1 - \omega^2/(\pi^2 \beta)) \\
 \sim {1 \over  \pi  \beta} {1 \over | \omega|} + {1 \over  \pi^3 \beta^2} | \omega| \log | \omega| + {\rm O}(|\omega|).
   \end{multline} 
   Here the symbol $ \overset{\cdot}{\sim} $ denotes the singular portion of the asymptotic expansion only. In the case $\beta = 2$, the asymptotic
   expansion (\ref{3.2g}) has been established rigorously \cite{RK23}, and moreover the next term in the asymptotic expansion has been determined to equal
   $|\omega|( {1 \over 24 \pi} - {1 \over 4 \pi^3} \log(2 \pi) )$. 
   
   The result  (\ref{3.2g}) has consequence  in the study of the covariances in (\ref{1.1x}).
   Previously   \cite{RTK23}  in the case $\beta = 2$ it has been shown in how to use (\ref{3.2g}) in (\ref{1.1x}) to deduce several terms in the large $k$ expansion
   of the covariance ${\rm cov}_\infty  (s_j(0), s_{j+k}(0))$. Leaving $\beta$ as a variable, the same calculation gives
  \begin{equation}\label{3.9}   
  {\rm cov}_\infty  (s_j(0), s_{j+k}(0)) \mathop{\sim}\limits_{k \to \infty} - {1 \over \beta \pi^2 k^2} - {6 \over \beta^2 \pi^4 k^4}
  \Big ( \log (2 \pi k) + \gamma - {11 \over 6} \Big ) + \cdots,
  \end{equation}
  where $\gamma$ denotes Euler's constant. In fact this asymptotic prediction can be validated numerically.
 For this one first recalls that the footnote associated with the paragraph containing (\ref{1.1x}) gives a formula for the LHS of (\ref{3.9}) in terms of certain variances.
 High precision computation of  these variances as required to compute the LHS up to $k=9$ is given
 in  \cite[Table 7 for $\beta = 1$, Table 8 for $\beta = 2$]{Bo10}), thus allowing the accuracy of (\ref{3.9}) to be validated in both
 the cases $\beta = 1$ and $\beta = 2$. For example, with $k=9$ and $\beta = 2$, the high precision
 data implies $  {\rm cov}_\infty  (s_j(0), s_{j+k}(0)) = -0.00063203\dots$ whereas the asymptotic formula (\ref{3.9}) gives the value $-0.00063196\dots$.
With $k=9$ and $\beta = 1$, the high precision data gives  $  {\rm cov}_\infty  (s_j(0), s_{j+k}(0)) = -0.001263\dots$ while the
asymptotic formula gives $-0.001276\dots$.

We turn our attention now to explicit functional forms for $ \mathcal E_{\infty, \beta}$.
For $\beta =1,2$ and 4, Fredholm determinant formulas for $ \mathcal E_{\infty, \beta} ((0,s); 1 - e^{i \omega}) $
following from (\ref{2.5x}) and (\ref{K1a}), and $\tau$ function formulas following from
(\ref{2.5d}) and (\ref{Sa1a}), are well known \cite[Ch.~8 \& 9]{Fo10}. These have already been presented for $\beta = 2$ in \cite{RK23}.
For future reference, we make note of the  Fredholm determinant formula in this case. Thus introduce the kernel $ K_\infty(x,y)$ and
associated integral operator $\mathcal K_\infty^{(0,s)}$ as in (\ref{2.5c}).
Then one has
 the fundamental formula in the theory of matrix ensembles with unitary symmetry in bulk scaling,
 as it applies to the generating
 function for the conditioned gap probabilities \cite[Eq.~(9.21)]{Fo10},
  \begin{equation}\label{3.2i}
 \mathcal E_{\infty, 2} ((0,s); 1 - e^{i \omega}) = \det \Big (\mathbb I - (1 - e^{i \omega}) \mathcal K_\infty^{(0,s)} \Big ).
  \end{equation}

To state the analogue of (\ref{3.2i}) for $\beta = 1$ and 4, and also its companion giving a $\tau$-function evaluation \cite{JMMS80},
some notation is needed. The first relates to the limit of the generating function for the conditioned gap
probabilities in the orthogonal group, where we set \cite[Eq.~(8.101)]{Fo10}
\begin{equation}\label{3.3}
 \mathcal E^{O^\pm}((0,s);\xi) = \lim_{N \to \infty}  \mathcal E^{O^\pm(2N+1)}((0,\phi);\xi) \Big |_{\phi = \pi s/N}.
\end{equation}
Next, in terms of (\ref{2.5c}) introduce the kernels
\begin{equation}\label{3.3a}
K_\infty^\pm(x,y) =  K_\infty(x,y)  \pm K_\infty(x,-y) , 
 \end{equation}       
 together with the corresponding integral operators on $(0,s)$, denoted   $\mathcal K_\infty^{\pm,(0,s)}$. For the $\tau$ function form,
 introduce the particular
 $\sigma$PIII$'$ transcendents $v = v_\pm(t;\xi)$ as the solution of the corresponding $\sigma$PIII$'$ differential equation
 \begin{equation}\label{3.3da}
 (t v'')^2 - \tfrac{1}{4}(v')^2 + v'(4 v' - 1) (v - t v') = 0,
  \end{equation}  
  subject to the boundary conditions
   \begin{equation}\label{3.3dm}
   v_-(t;\xi) \mathop{\sim}\limits_{t \to 0^+}  {\xi \over \pi} t^{1/2} + {\rm O}(t), \quad  v_+(t;\xi) \mathop{\sim}\limits_{t \to 0^+}  {\xi \over 3 \pi} t^{3/2} + {\rm O}(t).
    \end{equation} 
 These quantities are related by the equalities
 \begin{equation}\label{3.3d}
 \mathcal  E^{O^\pm}((0,s);\xi) = \det (\mathbb I - \xi K_\infty^{\pm,(0,s)}) = \exp \bigg ( - \int_0^{(\pi s)^2} v_\pm(t;\xi) \, {dt \over t} \bigg )
 \end{equation} 
 (here the first equality follows from (\ref{K1a}), while the second can be found in \cite[Eq.~(8.102), after correction of the signs on the LHS]{Fo10}).
 
 The relevance of the above results for $\beta = 1$ and 4 is immediate upon using (\ref{3.3}) to take the $N \to \infty$ limit
 in (\ref{2.7}) and (\ref{2.7S}).
 
 \begin{prop} \label{P3.2} (\cite[Eqns.~(8.152) and (8.159)]{Fo10}) We have\footnote{It is furthermore known that $ \mathcal  E_{\infty, 2} ((0,s); \xi)  =
 \mathcal  E^{O^+}((0,s/2);\xi)  \mathcal  E^{O^-}((0,s/2);\xi)$ \cite[Eq.~(8.129)]{Fo10}, which in view of the first equality in  (\ref{3.3d}) provides
 a Fredholm determinant evaluation distinct from (\ref{3.2i}).}
 \begin{align*}
 \mathcal  E_{\infty, 1} ((0,s); \xi) & = \frac{1 - \xi}{2 - \xi} \mathcal E^{O^+}((0,s/2);\xi(2 - \xi)) +  \frac{1}{2 - \xi} \mathcal E^{O^-}((0,s/2);\xi(2-\xi)), \\
  \mathcal  E_{\infty, 4} ((0,s); \xi) & = {1 \over 2} \Big (  \mathcal  E^{O^+}((0,s);\xi) +  \mathcal E^{O^-}((0,s);\xi) \Big ).
   \end{align*}
 \end{prop}
 
 \subsection{Numerical computation}\label{S3.2}
 In \cite{RK23}, both the Fredholm form (\ref{3.2i}), and its $\tau$ function companion, have been used to give an accurate numerical
 evaluation of $S_{\infty,2}(\omega)$, $0 < \omega < \pi$ (recall that the small $\omega$ form is already determined by 
 (\ref{3.2g}), furthermore supplemented by knowledge of explicit form of the O$(|\omega|)$ term). Here we will make use of
 the Fredholm forms implied by Proposition \ref{P3.2} in (\ref{3.2c}) to provide a numerical evaluation of $S_{\infty,\beta}(\omega)$
 for $\beta = 1, 4$.
 
 The optimal numerical procedure to compute a Fredholm determinant with analytic kernel $\xi K(x,y)$ is based on its
 definition in terms of eigenvalues (recall (\ref{2.5b}), although in the present setting the integral operator is no longer
 of finite rank) \cite{Bo08}. Thus one replaces the eigenvalue/ eigenfunction equation
  \begin{equation}\label{3.4}
  \xi \int_0^s K(x,y) \psi(y) \, dy = \lambda \psi(x)
 \end{equation} 
 by a system of $m$ linear equations
   \begin{equation}\label{3.4a}  
   \xi \sum_{k=1}^m w_j^{1/2} K(y_j,y_k) w_k^{1/2} \psi(y_k) \approx  \lambda \psi(y_j), \quad (j=1,\dots,m).
 \end{equation} 
 In this approximation the definition of a Fredholm determinant in terms of eigenvalues    (\ref{2.5b}) relates to
 the characteristic equation of the symmetric matrix $[ w_j^{1/2} K(y_j,y_k) w_k^{1/2} ]_{j,k=1}^m$ and
 gives 
   \begin{equation}\label{3.4b}  
   \det ( \mathbb I - \xi \mathbb K_s) \approx \det \Big [ \delta_{j,k} - \xi w_j^{1/2} K(y_j,y_k) w_k^{1/2} \Big ]_{j,k=1,\dots,m}.
 \end{equation}    
  An essential feature of the theory of \cite{Bo08} is that $\{y_j, w_j\}_{j=1}^m$ --- referred to as the abscissa and weights --- are to be chosen that $\int_0^s f(x) \, dx =
  \sum_{j=1}^m w_j(y_j) f(y_j)$ for all polynomials of degree less than $m^*$ (i.e.~$\{y_j, w_j\}_{j=1}^m$  correspond to an
  $m$-point quadrature rule of order $m^*$). The larger the value of $m$, the more accurate the approximation, which moreover
  is quantified in \cite{Bo08} in relation to (\ref{3.4b}).
  
  In Mathematica, the abscissa and weights can be extracted using a single  line of code, 
 e.g.
  \begin{equation}\label{3.4c}   
  {\tt \{absc, weights, errweights\} = NIntegrate`GaussBerntsenEspelidRuleData[n, prec]}
  \end{equation} 
   in relation to 
 Gaussian quadrature with $m=2n+1$, $m^* = 2m$. Here {\tt prec} is the precision, while {\tt errweights} are particular error
 weights which we do not require. Another convenient feature of Mathematica for present purposes
 is that the special function sinc$\,x$ is in-built (otherwise the issue of division by zero would
 be encountered for an evaluation at $x=0$).
 
 By the use of (\ref{3.4b}) in the case of the kernels (\ref{3.3a}), and considering
 (\ref{3.2c})  and  Proposition \ref{P3.2},
  we have available a computational
 scheme for  $S_{\infty,\beta}(\omega)$
 in the cases $\beta = 1, 4$. Actually, it is convenient to rescale the integral operators by the change
 of variables $x,y \mapsto sx,sy$. Then the new
 integral operator is on the interval $(0,1)$ and the new kernels are
  \begin{equation}\label{3.4d} 
  K_{\infty}^\pm(x,y;s) := K_{\infty}(x,y;s)  \pm K_{\infty}(x,-y;s), \quad  K_{\infty}(x,y;s) := s \, {\rm sinc}\, \pi s (x-y),
    \end{equation}    
    
    In the implementation of this procedure, it is found that for large $s$ the parameter $n$ in (\ref{3.4c}) needs to be taken
    as (approximately) $s$ to ensure accurate evaluation. Hence the speed of the evaluation decreases as $s$ increases. This
    presents a difficulty, as the slow oscillatory decay of the integrand, proportional to\footnote{This asymptotic form is evident from
    our numerical computations, although for $\omega$ near $\pi$ there is further modulated structure due to, most likely, a large
    amplitude associated with the known (for $\beta = 2$) oscillatory leading decaying correction term to the logarithm of (\ref{3.4e})
    \cite{KKMST09}; see Figure \ref{Fa5} for an example.}
  \begin{equation}\label{3.4e}  
  {\cos \omega s \over s^{\omega^2/(\pi^2 \beta)} }
     \end{equation}    
     as seen from the asymptotic formula (\ref{3.2f}), does not allow for a truncation of the integrand with negligible remainder. Fortunately there
     is a work around. First, the integrand was computed at close enough intervals between $0$ and $s^*$ (we found
     that a discretisation using a lattice of spacing $1/10$ was adequate, and for our test of the $\beta = 2$
     case we chose $s^*=100$). Then, as described in Section \ref{S2.2}, an interpolation function was formed, and the
     value of the integral in the range $(0,s^*)$ computed. Most importantly, with $s^*$ chosen so that $\cos \omega s^* = 1$
     (this detail was found to be essential for the stability of our numerics),
     we normalise (\ref{3.4e}) by 
       \begin{equation}\label{3.4m}  
     A_{\beta}(s^*;\omega) := (s^*)^{\omega^2/(\pi^2 \beta)} {\rm Re}  \, \mathcal E_{\infty,\beta}((0,s^*);1 - e^{i \omega}),
    \end{equation}    
     to give agreement with the value of the integrand at $s^*$,
     and take the value of the remainder due to the truncation as $ A_{\beta}(s^*;\omega)$ times
   \begin{equation}\label{3.4f}  
   \int_{s^*}^\infty        {\cos \omega s \over s^{\omega^2/(\pi^2 \beta)} } \, ds,
  \end{equation}  
 which itself is fast to compute using Mathematica. Here one first makes use of the equality in
 (\ref{3.2g}), for which the remaining task is to compute the integral over $(0,s^*)$.
 
  \begin{figure*}
\centering
\includegraphics[width=0.6\textwidth]{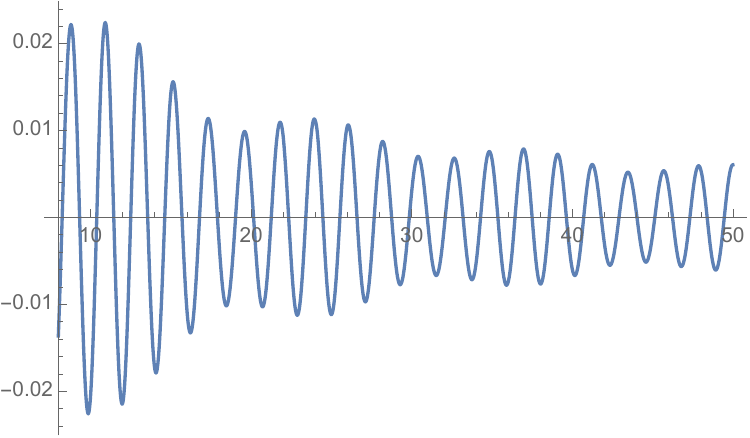}
\caption{Plot of  Re$\,  \mathcal E_{\infty, 2} ((0,s); 1 - e^{i \omega})$, for $s$ between 8 to 50
and with $\omega = 23 \pi/25$, exhibiting a modulated structure beyond the asymptotic form (\ref{3.4e}).}
\label{Fa5}
\end{figure*}
 
 The case $\beta = 2$ can be used as a  test bench for this procedure. This is due to the previously
 mentioned accurate (error no more than $10^{-6}$) tabulation  of $S_{\infty,2}(\omega)$
 in \cite{RK23}.\footnote{This was achieved using a combination of
 both its $\tau$-function evaluation, and that in terms of the Fredholm determinant (\ref{3.2i}) as $\omega$ approached $\pi$.
 In the former setting, a $\sigma$PV differential equation characterisation analogous to the second equality in (\ref{3.3d}) was used
 to achieve integration up to $s^*$ as large as $10^5$. When using the Fredholm determinant approach, for which the numerics becomes
 more costly as $s^*$ increases, values as large of $2 \times 10^4$ were achieved 
 using cluster (rather than desktop) computing. Throughout, the error estimates made use of (a corrected \cite[Eq.~(4.6)\&(4.7)]{RK23}) asymptotic expansion of the
 $\sigma$PV transcendent from \cite{CK15}.}
  We compare the results of our computation, making use
 too of the Fredholm determinants in (\ref{3.3d}) via the identity of the footnote associated with Proposition \ref{P3.2},
 in Table \ref{TA}. The displayed accuracy of 4 or 5 decimal places (typically improving as $\omega$ increases, which is to be
 expected as the exponent in (\ref{3.4e}) then increases while the amplitude $A_{\beta}(s^*;\omega)$ decreases) gives us confidence
 that our numerical approach is well founded.

  \vspace{.5cm}
  
 \begin{table}[ht]
  \centering
  \begin{tabular}{ll}
    \begin{tabular}{c|c|c}
$k$ &    $S_{\infty,2}^{\rm exact}(\omega)$ &  $S_{\infty,2}^{\rm approx}(\omega)$\\
    \hline & \\[-1.0em]
    1  & 0.629975 &   0.62992 \\
    2 & 0.312700 &  0.312651 \\
    3  & 0.207411& 0.207367 \\ 
   4 & 0.155436& 0.155395 \\
   5 &   0.124957 & 0.12492 \\
   6 & 0.105347 & 0.105322
    \end{tabular}\hfill
    &
    \begin{tabular}{c|c|c}
$k$ &       $S_{\infty,2}^{\rm exact}(\omega)$ &  $S_{\infty,2}^{\rm approx}(\omega)$\\
    \hline  & \\[-1.0em]
    7 & 0.092046  &  0.092037 \\
    8 &   0.082787 & 0.082789 \\  
    9  & 0.076329  & 0.076323 \\  
   10 & 0.071951  & 0.071945 \\  
   11 &  0.069228 & 0.069228 \\
   12 &  0.067922 & 0.067925
    \end{tabular}
     \end{tabular}
      \caption{Here $\omega = {2 \pi k/25}$, $S_{\infty,2}^{\rm exact}(\omega)$ are the values as calculated in \cite{RK23}
      in the case $\beta = 2$,
      with accuracy better than 6 decimals stated above, while $S_{\infty,2}^{\rm approx}(\omega)$ is
      the result of the present computational scheme with $s^*=100$. }
         \label{TA}
 \end{table}%
 
 The use of the  Fredholm determinants in (\ref{3.3d}) via the identity of the footnote associated with Proposition \ref{P3.2}
 in our testing of the proposed numerical scheme in the case $\beta = 2$ has, according to the first formula in Proposition
 \ref{P3.2}, the feature that the corresponding interval length is $s/2$. Recalling our findings reported above (\ref{3.4e})
 regarding the relation between the interval length and $n$ in (\ref{3.4c}), this implies that the computational cost for the
 $\beta = 1$ case of (\ref{3.2c}) using our scheme is essentially the same as for the $\beta = 2$ case just carried out.
 Moreover, there is no reason to expect that the results obtained to be any less accurate (four to five decimals, improving
 as $\omega$ increases). The results of our computation are tabulated in Table \ref{TA1}.
 
   \vspace{.5cm}
  
 \begin{table}[ht]
  \centering
  \begin{tabular}{llll}
    \begin{tabular}{c|c}
$k$ &    $S_{\infty,1}^{\rm approx}(\omega)$ \\
    \hline  & \\[-1.0em]
    ${1 \over 2}$  & 2.47539 \\
    1 &  1.2072 \\
    ${3 \over 2}$  &  0.784986 \\ 
   2 &   0.574758\\
   ${5 \over 2}$ &    0.449517 \\
   3 &  0.366891
    \end{tabular}\hfill
    &
    \begin{tabular}{c|c}
$k$ &       $S_{\infty,1}^{\rm approx}(\omega)$ \\
    \hline  & \\[-1.0em]
    ${7 \over 2}$ &    0.308633 \\
    4 &    0.26563 \\  
    ${9 \over 2}$  &  0.232872 \\  
   5 &   0.207143 \\  
   ${11 \over 2}$ &   0.18673 \\
   6 &  0.170155
    \end{tabular}\hfill
    &
    \begin{tabular}{c|c}
$k$ &       $S_{\infty,1}^{\rm approx}(\omega)$ \\
    \hline  & \\[-1.0em]
    ${13 \over 2}$ &  0.156537  \\
    7 &     0.145469 \\  
    ${15 \over 2}$  &   0.136032  \\  
   8 &   0.128457  \\  
   ${17 \over 2}$ &  0.121966  \\
   9 &  0.116653
     \end{tabular}\hfill
    &  
    \begin{tabular}{c|c}
$k$ &       $S_{\infty,1}^{\rm approx}(\omega)$ \\
    \hline & \\[-1.0em]
    ${19 \over 2 }$ &  0.112358  \\
    10 &  0.108769   \\  
    ${21 \over 2 }$ &   0.106038 \\  
   11 &    0.103944 \\  
   ${23 \over 2}$ &  0.102462  \\
   12 &  0.101619
    \end{tabular}
     \end{tabular}
      \caption{Here $\omega = {2 \pi k/25}$ while $S_{\infty,1}^{\rm approx}(\omega)$ is
      the result of our  computational scheme applied to the exact expression in
      Proposition \ref{P3.2} in the case $\beta = 1$. }
         \label{TA1}
 \end{table}%

 In distinction to the Fredholm determinant formulations for the evaluation of $S_{\infty,\beta}(\omega)$ used above
 for $\beta = 2$ and 1, according to the second  formula in Proposition  \ref{P3.2} the interval length in the
 Fredholm determinant formulation for $\beta = 4$ is $s$ rather than $s/2$. This means that for the same computational
 time, we can now only compute $  \mathcal  E_{\infty, 4} ((0,s); \xi)$ for $s$ between 0 and $s^*=50$, rather than
 between 0 and $s^*=100$ as before. One notes that for our range of $\omega = \pi k /25$ ($k=1,\dots,24$), it remains true that
 $s^* \omega$ is an integer, which we have previously remarked appears crucial to the success of our
 approach. To quantify the effect of this reduction, we first repeated the
 $\beta = 2$ computation, now restricted to $s^*=50$. Our results, presented in Table \ref{TA2}, show that the accuracy
 is barely changed to what was found when using $s^*=100$. This same effect was observed in
 relation to the $\beta = 1$ data of Table \ref{TA1} when reducing down to $s^*=50$. Thus the accuracy is 
 expected in relation to the results of our numerical scheme for the computation of
 $S_{\infty,\beta}(\omega)$ with $\beta = 4$, presented in Table \ref{TA3}, are expected to be comparable to those
 for $\beta = 1$, notwithstanding the smaller value of $s^*$.

   \vspace{.5cm}
  
 \begin{table}[ht]
  \centering
  \begin{tabular}{llll}
    \begin{tabular}{c|c}
$k$ &    $S_{\infty,2}^{* \, \rm approx}( \omega)$ \\
    \hline & \\[-1.0em]
    1 &   0.629870 \\
   2 &  0.312602 \\
   3 &  0.207322
    \end{tabular}\hfill
    &
    \begin{tabular}{c|c}
$k$ &       $S_{\infty,2}^{* \, \rm approx}(\omega)$ \\
    \hline & \\[-1.0em]
    4 &   0.155354 \\  
   5 &     0.124886 \\ 
   6 &  0.105293
    \end{tabular}\hfill
    &
    \begin{tabular}{c|c}
$k$ &       $S_{\infty,2}^{* \,\rm approx}(\omega)$ \\
    \hline & \\[-1.0em]
    7 &    0.092012 \\ 
   8 &    0.082765 \\ 
   9 &  0.076309
     \end{tabular}\hfill
    &
    \begin{tabular}{c|c}
$k$ &       $S_{\infty,2}^{* \, \rm approx}(\omega)$ \\
    \hline & \\[-1.0em]
    10 &   0.071934 \\  
   11 &    0.069219 \\ 
   12 & 0.067917
    \end{tabular}
     \end{tabular}
      \caption{Here $\omega = {2 \pi k/25}$ while $S_{\infty,2}^{* \,\rm approx}(\omega)$ is
      the result of our  computational scheme with $s^* = 50$ in the case $\beta = 2$. }
         \label{TA2}
 \end{table}%

    \vspace{.5cm}
  
 \begin{table}[ht]
  \centering
  \begin{tabular}{llll}
    \begin{tabular}{c|c}
$k$ &    $S_{\infty,4}^{* \,  \rm approx}(\omega)$ \\
    \hline & \\[-1.0em]
    1 &   0.322125 \\
   2 & 0.163706\\
   3  & 0.111115  
    \end{tabular}\hfill
    &
    \begin{tabular}{c|c}
$k$ &       $S_{\infty,4}^{* \, \rm approx}(\omega)$ \\
    \hline & \\[-1.0em]
    4 &   0.085100 \\  
   5 &    0.069785 \\ 
   6 & 0.059908
    \end{tabular}\hfill
    &
    \begin{tabular}{c|c}
$k$ &       $S_{\infty,4}^{* \,\rm approx}(\omega)$ \\
    \hline & \\[-1.0em]
    7 &  0.053163 \\ 
   8 &   0.048451 \\ 
   9 &  0.045171
     \end{tabular}\hfill
    &
    \begin{tabular}{c|c}
$k$ &       $S_{\infty,4}^{* \, \rm approx}(\omega)$ \\
    \hline & \\[-1.0em]
    10 &   0.042899 \\  
   11 &     0.041530 \\ 
   12 &   0.040840
    \end{tabular}
     \end{tabular}
      \caption{Here $\omega = {2 \pi k/25}$ while $S_{\infty,4}^{* \,\rm approx}(\omega)$ is
      the result of our  computational scheme with $s^* = 50$ in the case $\beta = 4$. }
         \label{TA3}
 \end{table}%
 
 We finish our discussion relating to the Fredholm determinant approach with some remarks relating to the approach to the large $N$ limit
 $S_{\infty,\beta}(\omega)$ via a sequence of increasing $N$ values of
 $S_{N,\beta}(\omega)$. Some oscillatory behaviour relating to this is already seen
 in Figures \ref{Fa2}, \ref{Fa3} and \ref{Fa4}. However, further investigation shows
monotonic convergence if the sequence of $N$ values is chosen such
 that $N \omega$ is a multiple of $2 \pi$. Some data is given in Table \ref{TA4},
 which is to be compared against the corresponding data in Table \ref{TA2}.
 The rate of convergence appears to be O$(1/N^2)$, with an $\omega$ dependent
 proportionality. Such a rate is known previously for the spacing distributions of the
 circular ensembles \cite{BFM17}.
 The stability of the numerics in this circumstance is analogous to our finding in
 relation to the computation of $S_{\infty,\beta}(\omega)$ with the criterium that
  $s^* \omega$ be a multiple of $2 \pi$ (recall the text above (\ref{3.4m})).

  \begin{table}[ht]
  \centering
  \begin{tabular}{llll}
    \begin{tabular}{c|c|c}
$k$ &    $S_{50,1}(\omega)$ &  $S_{100,1}(\omega)$ \\
    \hline  & \\[-1.0em]
    ${1 \over 2}$  & 2.47335  & 2.47454 \\[.2em]
 ${3 \over 2}$     & 0.783343 & 0.784283 \\[.2em]
    ${5 \over 2}$  & 0.448083 & 0.448897 \\ [.2em]
 ${7 \over 2}$   &  0.30734 &  0.30806 \\[.2em]
   ${9 \over 2}$ & 0.231662  &0.232307 \\[.2em]
    ${11 \over 2}$ & 0.185616  & 0.186201
    \end{tabular}\hfill
    &
    \begin{tabular}{c|c|c}
$k$ &       $S_{50,1}(\omega)$  &  $S_{100,1}(\omega)$\\
    \hline  & \\[-1.0em]
    ${13 \over 2}$ &  0.15556 &  0.156097 \\[.2em]
   ${15 \over 2}$   & 0.135148  &  0.135649 \\ [.2em] 
    ${17 \over 2}$  & 0.121063 & 0.121535 \\  [.2em]
 ${19 \over 2}$   &  0.111435  & 0.111886 \\ [.2em] 
   ${21 \over 2}$ & 0.10517 &  0.105606 \\ [.2em]
 ${23 \over 2}$  &   0.101632 &   0.102059
       \end{tabular}
     \end{tabular}
 \caption{Here $\omega = {2 \pi k/25}$ while $S_{N,1}(\omega)$ is
      computed according the procedure of Section \ref{S2.3}.}
         \label{TA4}
 \end{table}%
 
 As an alternative to the Fredholm determinant evaluation, we know from (\ref{3.3d}) that the
 key quantities $\mathcal E^{O^\pm}((0,s);\xi)$ permit a $\sigma{\rm PIII}'$ $\tau$ function 
 evaluation. We know too from the experience of \cite{RK23} as reported in the footnote associated
 with the paragraph below (\ref{3.4f}) that, provided $\omega$ is not close to $\pi$,  this is well suited to extending the numerics to larger
 $s^*$, and so improving their accuracy. Relevant here is the large $t$ form of the $\sigma {\rm PIII}'$ transcendent
 $v_\pm(t;\xi)$ in  (\ref{3.3d}), as there is still the task of estimating the error due to the eventual truncation at $s^*$. This is of independent
 interest as the connection problem for the small $t$ behaviour (\ref{3.3dm}) and its asymptotic form for $t \to \infty$, which is considered
 as fundamental in the theory of Painlev\'e systems. References addressing this problem in the case of $\sigma {\rm PIII}'$
  include \cite{Ki89,LS01}. However, there is no literal transcription from these works to our setting. We leave this
 problem for future research. For the time being, we put forward (checked by numerical comparisons)  a  conjecture.
 
 \begin{con}\label{C3.1}
 Let $\epsilon = \pm $, $\tilde{\omega} = \omega/2 \pi$ and specify $v_\epsilon(t;\xi)$ as the particular $\sigma {\rm PIII}'$ transcendent
 satisfying the differential equation (\ref{3.3da}) subject to   the $ t \to 0^+$  boundary condition (\ref{3.3dm}).
For $ 0 \le \omega < \pi$, up to a yet to be determined order, one has
   \begin{equation}\label{c.1}  
   v_\epsilon(t;\xi)|_{\xi = 1-e^{i \omega}} \mathop{\sim}\limits_{t \to \infty} {A_\epsilon(t;\omega) \over B_\epsilon(t;\omega) }
   \end{equation}
   where
   \begin{multline}\label{Ae}
 A_\epsilon(t;\omega)  = - i \tilde{\omega} \Big ( \sqrt{t} - {\sin^2(2\sqrt{t}) \over 4 \sqrt{t}} \tilde{\omega} \Big ) + {1 \over 2} \tilde{\omega}^2 \\
 -  {\epsilon (1 - \tilde{\omega}) \over 4 \pi} \Big (  \sin(\pi  \tilde{\omega} ) \Gamma^2(1 - \tilde{\omega}) (4 \sqrt{t})^{2  \tilde{\omega} } e^{-2i \sqrt{t}} + \overline{(  \tilde{\omega} \mapsto -  \tilde{\omega})} \Big )
 \end{multline}
 and
  \begin{equation} 
  B_\epsilon(t;\omega) = 1 + i {\epsilon \over \pi} \Big(  \sin(\pi \tilde{\omega}) \Gamma^2(1 - \tilde{\omega}) (4 \sqrt{t})^{2  \tilde{\omega} - 1} e^{-2i \sqrt{t}} + \overline{(  \tilde{\omega} \mapsto -  \tilde{\omega})} \Big).
  \end{equation} 
   \end{con}

\begin{figure*}
\centering
\includegraphics[width=100mm,height=85mm]{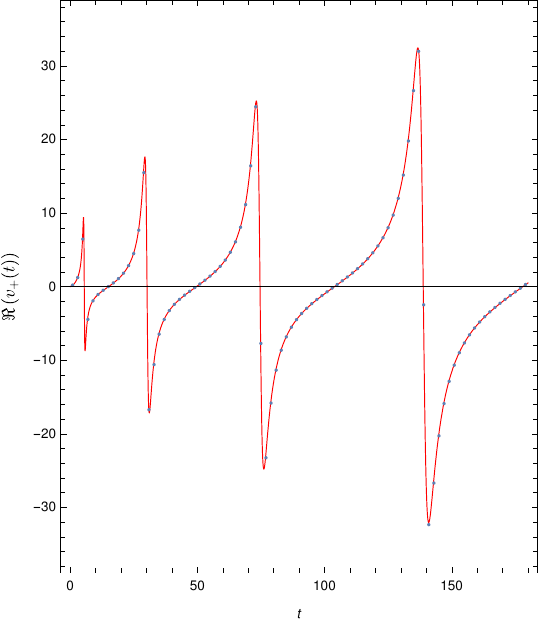}
\caption{Plot of  Re$\,  v_+ (t; 1 - e^{i \omega})$, for $t$ between 0 to 180
and with $\omega = 31 \pi/32$. Here the solid red line is using the proposed asymptotic form (\ref{c.1}), while the dots are the values obtained by a numerical solution of  (\ref{3.3da}) subject to the $ t \to 0^+$  boundary condition (\ref{3.3dm}).}
\label{Fa6}
\end{figure*}  
   
 \begin{remark}
 As $\omega \to \pi^-$, the right hand side of (\ref{c.1}) reduces to
    \begin{equation}\label{c.2}  
 {i \over 2} \sqrt{t} {e^{-2 i \sqrt{t}} + \epsilon i \over e^{-2 i \sqrt{t}} - \epsilon i }.
   \end{equation} 
This displays an infinite number of poles along the positive real $t$-axis, which in turn correspond to zeros of $\mathcal E^{O^\pm}((0,s);\xi)|_{\xi = 2}$.
Although $ v_\epsilon(t;\xi)|_{\xi = 1-e^{i \omega}}$ is analytic on the positive   real $t$-axis for $ 0 \le \omega < \pi$, numerical plots for
$\omega $ less than but close to $\pi$ show that its gradient
changes rapidly at the poles of (\ref{c.2}), and that its magnitude is large; see Figure \ref{Fa6} for an example. Due to this effect, numerical
computation for very large $t$ becomes unstable, in contrast to numerical
computation of  $ v_\epsilon(t;\xi)|_{\xi = 1-e^{i \omega}}$ with smaller values of $\omega$. 
\end{remark}

  \subsection*{Acknowledgements}
	This work  was supported
	by the Australian Research Council 
	 Discovery Project grant DP210102887.
	 Helpful feedback on the early draft by E.~Kanzieper
	 is acknowledged, as is correspondence with R.~Riser
	 providing the information reported in the footnote
	 associated with the paragraph below (\ref{3.4f}).

\providecommand{\bysame}{\leavevmode\hbox to3em{\hrulefill}\thinspace}
\providecommand{\MR}{\relax\ifhmode\unskip\space\fi MR }
\providecommand{\MRhref}[2]{%
  \href{http://www.ams.org/mathscinet-getitem?mr=#1}{#2}
}
\providecommand{\href}[2]{#2}

  \end{document}